\documentclass[a4paper,UKenglish,cleveref, autoref, thm-restate, numberwithinsect]{lipics-v2021}

\usepackage{scalerel}
\usepackage{proof}
\usepackage{ifthen}
\usepackage{mathtools}
\usepackage{tikz}
\usetikzlibrary{positioning,arrows,calc}

\newboolean{OmitProofs}
\setboolean{OmitProofs}{true} 

\newcommand{\m}[1]{\mathsf{#1}}
\newcommand{\mr}[1]{\mathrel{#1}}
\newcommand{\seq}[2][n]{{#2_1},\dots,{#2_{#1}}}
\newcommand{\sig}[2][n]{{#2_1}\times\cdots\times{#2_{#1}}}
\newcommand{\h}[1][.3]{\hspace{#1mm}}
\newcommand{\SET}[1]{\{\h#1\h\}}
\newcommand{\inter}[1]{[\![{#1}]\!]}

\newcommand{\CO}[1]{[\h#1\h]} 

\newcommand{\xE}{\mathcal{E}}
\newcommand{\xF}{\mathcal{F}}

\newcommand{\xI}{\mathcal{I}}
\newcommand{\xJ}{\mathcal{J}}
\newcommand{\xM}{\mathcal{M}}
\newcommand{\xR}{\mathcal{R}}
\newcommand{\xS}{\mathcal{S}}
\newcommand{\xT}{\mathcal{T}}

\newcommand{\xV}{\mathcal{V}}
\newcommand{\xFTe}{\xF_{\m{te}}}
\newcommand{\xFTh}{\xF_{\m{th}}}
\newcommand{\xSTe}{\xS_{\m{te}}}
\newcommand{\xSTh}{\xS_{\m{th}}}
\newcommand{\xVTe}{\xV_{\m{te}}}
\newcommand{\xVTh}{\xV_{\m{th}}}
\newcommand{\SigmaTh}{\Sigma_{\m{th}}}
\newcommand{\SigmaTe}{\Sigma_{\m{te}}}

\newcommand{\fI}{\mathfrak{I}}
\newcommand{\fJ}{\mathfrak{J}}
\newcommand{\fM}{\mathfrak{M}}
\newcommand{\fT}{\mathfrak{T}}

\newcommand{\fU}{\mathfrak{U}}

\newcommand{\xRca}{\xR_\m{calc}}
\newcommand{\Val}{\xV\m{al}}
\newcommand{\Var}{\xV\m{ar}}

\newcommand{\LVar}{\mathcal{L}\Var}

\newcommand{\Dom}{\mathcal{D}\m{om}}

\newcommand{\VDom}{\mathcal{V}\Dom}
\newcommand{\Pos}{\mathcal{P}\m{os}}

\newcommand{\sort}[1]{\m{#1}}
\newcommand{\Bool}{\sort{Bool}}
\newcommand{\Int}{\sort{Int}}

\newcommand{\CEqn}[4]{%
\ifthenelse{\equal{#1}{}}{%
\ifthenelse{\equal{#4}{}}{#2 \approx #3}{#2 \approx #3~\CO{#4}}}{%
\ifthenelse{\equal{#4}{}}{\Pi #1.\, #2 \approx #3}{\Pi #1.\, #2 \approx #3~\CO{#4}}%
}}

\newcommand{\R}{\rightarrow}
\newcommand{\Ra}[1][]{\R^{#1}}
\newcommand{\Rb}[1][]{\R_{#1}}
\newcommand{\Rca}{\Rb[\m{calc}]}
\newcommand{\Rru}{\Rb[\m{rule}]}

\renewcommand{\L}{\leftarrow}

\newcommand{\Lb}[1][]{\mr{\vphantom{\R}_{#1}{\L}}}
\newcommand{\Lca}{\Lb[\m{calc}]}

\newcommand{\Rbase}[1][]{\R_{\mathsf{base}}}

\newcommand{\Ca}[1][]{\xleftrightarrow{#1}}

\newcommand{\Cb}[1][]{\Ca[]_{#1}}
\newcommand{\Cab}[2][]{\Ca[#2]_{#1}}

\newcommand{\Cca}{\Cb[\m{calc}]}
\newcommand{\Cru}[1][\xE]{\ifthenelse{\equal{#1}{}}{\Cb[\m{rule}]}{\Cb[\m{rule},#1]}}
\newcommand{\Cbase}[1][\xE]{\ifthenelse{\equal{#1}{}}{\Cb[\m{base}]}{\Cb[\m{base},#1]}}

\newcommand{\cec}{\mathrel{\models_{\m{cec}}}}

\renewcommand{\ge}{\geqslant}
\renewcommand{\le}{\leqslant}

\newcommand{\Bfnum}[1]{\textcolor{darkgray}{\sffamily\bfseries #1}}

\title{Equational Theories and Validity for Logically Constrained Term Rewriting}

\author{Takahito Aoto}{Niigata University, Japan}{aoto@ie.niigata-u.ac.jp}
{https://orcid.org/0000-0003-0027-0759}
{JSPS KAKENHI Grant Numbers 21K11750, 24K14817.}

\author{Naoki Nishida}{
Nagoya University, Japan}{nishida@i.nagoya-u.ac.jp}{https://orcid.org/0000-0001-8697-4970}{JSPS KAKENHI Grant Number 24K02900.}

\author{Jonas Sch\"opf}{University of Innsbruck, Austria}{jonas.schoepf@uibk.ac.at}{https://orcid.org/0000-0001-5908-8519}{FWF (Austrian Science Fund) project I~5943-N.}

\authorrunning{T. Aoto, N. Nishida and J. Sch\"opf} 

\Copyright{Takahito Aoto, Naoki Nishida and Jonas Sch\"opf} 



\category{} 

\relatedversiondetails[cite=ANS24-arxiv]{Full Version}{https://arxiv.org/abs/2405.01174}

\funding{This research was supported by the FWF (Austrian Science Fund)
project I~5943-N and JSPS-FWF Grant Number JPJSBP120222001.}

\acknowledgements{We thank the anonymous reviewers for their valuable feedback, which improved the
paper.}

\nolinenumbers 

\EventEditors{}
\EventNoEds{2}
\EventLongTitle{}
\EventShortTitle{}
\EventAcronym{}
\EventYear{}
\EventDate{}
\EventLocation{}
\EventLogo{}
\SeriesVolume{}
\ArticleNo{}

\begin{document}

\maketitle

\begin{abstract}
Logically constrained term rewriting is a relatively new formalism where rules
are equipped with constraints over some arbitrary theory. Although there are
many recent advances with respect to rewriting induction, completion, complexity
analysis and confluence analysis for logically constrained term rewriting, these
works solely focus on the syntactic side of the formalism lacking detailed
investigations on semantics. In this paper,  we investigate a semantic side of
logically constrained term rewriting. To this end, we first define constrained
equations, constrained equational theories and validity of the former based on
the latter. After presenting the relationship of validity and conversion of
rewriting, we then construct a sound inference system to prove validity of
constrained equations in constrained equational theories. Finally, we give an
algebraic semantics, which enables one to establish invalidity of constrained
equations in constrained equational theories. This algebraic semantics derives a
new notion of consistency for constrained equational theories.
\end{abstract}

\section{Introduction}
\label{sec:intro}

Logically constrained term rewriting is a relatively new formalism building upon
many-sorted term rewriting and built-in theories. The rules of a logically
constrained term rewrite system (LCTRS, for short) are equipped with constraints
over some arbitrary theory, which have to be fulfilled in order to apply
rules in rewrite steps. This formalism intends to live up with data
structures which are often difficult to represent in basic rewriting, such
as integers and bit-vectors, with the help of external provers and their
built-in theories.

Logical syntax and semantics are often conceived as two sides of the same coin.
This is not exceptional, especially for equational logic in which term
rewriting lies. On the other hand, although there are many recent advances in
rewriting induction~\cite{FKN17tocl}, completion~\cite{WM18}, complexity
analysis~\cite{WM21}, confluence analysis~\cite{KN13frocos,NW18vstte,SM23} and
(all-path) reachability~\cite{CL18,KN23jlamp,KN23padl} for LCTRSs, these works
solely focus on the syntactic side of the formalism, lacking detailed
investigations on semantics.

In this paper, we investigate a semantic side of the LCTRS formalism. To
this end, we first define \emph{constrained equations} (CEs, for short) and
\emph{constrained equational theories} (CE-theories, for short). In
(first-order) term rewriting, the equational version of rewrite rules is
obtained by removing the orientation of the rules. However, in the case of
LCTRSs, if we consider a constrained rule $\ell \R r~\CO{\varphi}$ and relate
this naively to a CE $\ell \approx r~\CO{\varphi}$, which does not distinguish
between left- and right-hand sides, we lose information about the
restriction on the possible instantiation of variables. This motivates us to add
an explicit set $X$ to each CE $\ell \approx r~\CO{\varphi}$ as
$\CEqn{X}{\ell}{r}{\varphi}$\footnote{In the literature, some other approaches
exist. The computation of critical pairs is also prone of losing information
\cite{SM23}. They solved it by adding dummy constraints $x = x$ to the critical
pair. Another approach was proposed in~\cite{WM18} where $\LVar(\ell \approx
r~\CO{\varphi})$ was simply defined as $\Var(\varphi)$.}---we name variables in
$X$ as \emph{logical variables} with respect to the equation. A CE-theory
is then defined as a set of CEs. Similar to the rewrite steps of
LCTRSs, we define validity by convertibility if all logical variables are
instantiated by values---we denote this notion of validity as
\emph{CE-validity} for clarity.

After establishing fundamental properties of the CE-validity, we present
its relation to the conversion of rewriting. However, the conversion of rewriting is
useful in general to establish the validity of arbitrary CEs. This
motivates us to introduce $\mathbf{CEC}_0$, an inference calculus for
deriving valid CEs. After demonstrating the usefulness of $\mathbf{CEC}_0$ via
some derivations, we present a soundness theorem for the calculus. We also show
a partial completeness result, followed by a discussion why our system seems
incomplete. Afterwards we consider the opposite question, namely
how to prove that a CE is not valid for a particular CE-theory. To this
end, we introduce an algebraic semantics that captures CE-validity. We
give a natural notion of models for CE-theory, which we call \emph{CE-algebras}.
We establish soundness and completeness with respect to CE-validity for
this.

Figure~\ref{fig:overview} presents the relationships between the
introduced notions and results of this paper.
The following concrete contributions are covered in this paper:
\begin{enumerate}
    \item We propose a formulation of CEs and CE-theories.
    \item On top of that we devise a notion of validity of a CE for a
    CE-theory $\xE$, which we call CE-validity.
    \item We give a proof system $\mathbf{CEC}_0$, and show soundness
    (Theorem~\ref{thm:soundness of the system CEC0}) and a partial completeness
    result (Theorem~\ref{thm: partial completeness of CEC0}) with respect to
    CE-validity.
    \item We give a notion of CE-algebras and based on it we define
    algebraic semantics, which is sound (Theorem~\ref{thm:soundness of algebraic
    semantics w.r.t. constrained equational validity}) and complete
    (Theorem~\ref{thm:completeness}) with respect to CE-validity for
    \emph{consistent} CE-theories.
\end{enumerate}

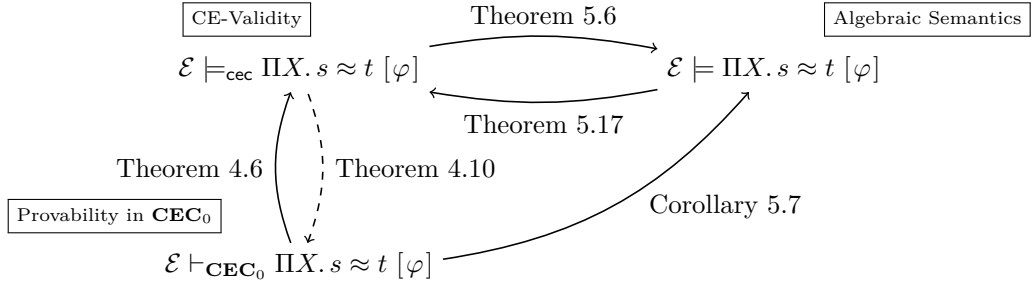
\begin{figure}[t]
    \centering
  \begin{tikzpicture}[node distance=2cm and 3cm,semithick]
  \node (CEC) {$\xE \cec \CEqn{X}{s}{t}{\varphi}$};
  \node (CECtitle) [above left = 0pt and -2cm of CEC] {\fbox{\scriptsize CE-Validity}};
  \node (CEC0) [below = of CEC] {$\xE \vdash_{\mathbf{CEC}_0} \CEqn{X}{s}{t}{\varphi}$};
  \node (CEC0title) [above left = 0pt and -1cm of CEC0] {\fbox{\scriptsize Provability in $\mathbf{CEC}_0$}};
  \node (Sem)  [right = of CEC] {$\xE \models \CEqn{X}{s}{t}{\varphi}$};
  \node (SEMtitle) [above right = 0pt and -1cm of Sem] {\fbox{\scriptsize Algebraic Semantics}};
  \draw[->] (CEC) 
        edge [bend left=10] 
        node [above = 1pt] {Theorem~\ref{thm:soundness of algebraic semantics w.r.t. constrained equational validity}} 
        (Sem);
  \draw[->] (Sem) 
      edge[bend left=10] 
      node [below = 1pt] {Theorem~\ref{thm:completeness}} 
      (CEC);
  \draw[<-] (CEC) 
      edge[bend right=20] 
      node [left = 1pt] {Theorem~\ref{thm:soundness of the system CEC0}} 
      (CEC0);
  \draw[<-,dashed] (CEC0) 
      edge[bend right=20] 
      node [right = 1pt] {Theorem~\ref{thm: partial completeness of CEC0}}
      (CEC);
  \draw[->] (CEC0)
      edge[bend right=20] 
      node [right = 10pt] {Corollary~\ref{cor:soundness of CEC0 wrt algebraic semantics}}
      (Sem);
\end{tikzpicture}
\caption{An overview of the main results of this paper.}
\label{fig:overview}
\end{figure}

We want to discuss some highlights of the last item for readers who are
familiar with algebraic semantics of equational logic. First of all, our
definition of CE-algebras admits extended underlying models, contrast to those
that precisely contain the same underlying models; we will demonstrate why this
generalization is required to obtain the completeness result. To reflect this
definition, it was necessary to modify the definition of congruence
relation to a non-standard one. Also, the notion of consistency with respect
to values arises to guarantee this modified notion of congruence in the term
algebras. Moreover, it also turns out that value-consistency is equivalent
to a more intuitive notion of consistency.

The remainder of the paper is organized as follows. In the next section, we
briefly explain the LCTRS formalism, and present some basic lemmas that are
necessary for our proofs. Section~\ref{sec:constrained-equational-validity}
introduces the notion of CEs, CE-theories and CE-validity, and presents basic
properties on CE-validity and its relation to the conversion of rewriting.
Section~\ref{sec:proving-validity} is devoted to our inference system
$\mathbf{CEC}_0$, including its soundness and partial completeness with respect
to CE-validity. In Section~\ref{sec:algebraic-semantics}, we present algebraic
semantics, and soundness and completeness results with respect to
CE-validity. Before concluding this paper in Section~\ref{sec:conclusion},
we briefly describe related work in Section~\ref{sec:related-work}. 
We provide only brief proof sketches of selected results in the main text.
However, all detailed proofs are given in the appendix.

\section{Preliminaries}
\label{sec:preliminaries}

In this section, we briefly recall LCTRSs~\cite{KN13frocos,FKN17tocl,SM23}.
Familiarity with the basic notions on mathematical
logic~\cite{EFTF94,logicandstructure} and term rewriting~\cite{BN98,Ohl02} is
assumed.

The (sorted) signature of an LCTRS is given by the set $\xS$ of sorts and the
set $\xF$ of $\xS$-sorted function symbols. Each $f \in \xF$ is equipped with a
sort declaration $f\colon \sig{\tau} \to \tau_0$ with $\tau_0,\ldots,\tau_n \in
\xS$; $\sig{\tau} \to \tau_0$ is said to be the sort of $f$, and we denote by
$\xF^{\sig{\tau} \to \tau_0}$ the set of function symbols of sort $\sig{\tau}
\to \tau_0$. For constants of sort $\to \tau$ we drop $\to$ and write $\tau$
instead of $\to \tau$. The set of $\xS$-sorted variables is denoted by $\xV$ and
the set of $\xS$-sorted terms over $\xF,\xV$ is $\xT(\xF,\xV)$. For each $\tau
\in \xS$, we denote by $\xV^{\tau}$ the set of variables of sort $\tau$ and by
$\xT(\xF,\xV)^\tau$ the set of terms of sort $\tau$; we also write $t^{\tau}$
for a term $t$ such that $t \in \xT(\xF,\xV)^\tau$. The set of variables
occurring in a term $t \in \xT(\xF, \xV)$ is denoted by $\Var(t)$ and can be
restricted by a set of sorts $T$ with $\Var^T(t) = \SET{x^{\tau} \in \Var(t)
\mid \tau \in T}$. A substitution $\sigma$ is a mapping $\xV \to \xT(\xF,\xV)$
such that $\Dom(\sigma) = \SET{ x \in \xV \mid x \neq \sigma(x)}$ is finite and
$\sigma(x^\tau) \in \xT(\xF,\xV)^\tau$ is satisfied for all $x \in
\Dom(\sigma)$.

In the LCTRS formalism, sorts are divided into two categories, that is, each
sort $\tau \in \xS$ is either a \emph{theory sort} or a \emph{term sort}, where
we denote by $\xSTh$ the set of theory sorts and by $\xSTe$ the set of term
sorts, i.e.\ $\xS = \xSTh \uplus \xSTe$. Accordingly, the set of variables is
partitioned as $\xV = \xVTh \uplus \xVTe$ by letting $\xVTh$ for the set of
variables of sort $\tau \in \xSTh$ and $\xVTe$ for the set of variables of sort
$\tau \in \xSTe$. Furthermore, we assume each function symbol $f \in \xF$ is
either a theory symbol or a term symbol, where all former symbols $f\colon
\sig{\tau} \to \tau_0$ need to satisfy $\tau_i \in \xSTh$ for all $0 \leqslant i
\leqslant n$. The sets of theory and term symbols are denoted by $\xFTh$ and
$\xFTe$, respectively: $\xF = \xFTh \uplus \xFTe$. Throughout the paper, we
consider signatures consisting of four components $\langle
\xSTh,\xSTe,\xFTh,\xFTe \rangle$. In some cases term/theory signature stands
for the two respective term/theory components of such a signature.

An LCTRS is also equipped with a model over the sorts $\xSTh$ and the symbols
$\xFTh$, which is given by $\xM = \langle \xI, \xJ \rangle$, where $\xI$ assigns
each $\tau \in \xSTh$ a \emph{non-empty} set $\xI(\tau)$, specifying its domain,
and $\xJ$ assigns each $f\colon \sig{\tau} \to \tau_0 \in \xFTh$ an
interpretation function $\xJ(f)\colon \xI(\tau_1) \times \cdots \times
\xI(\tau_n) \to \xI(\tau_0)$. In particular, $\xJ(c) \in \xI(\tau)$ for any
constant $c \in \xFTh^{\tau}$. We suppose for each $\tau \in \xSTh$, there
exists a subset $\Val_\tau \subseteq \xFTh^{\tau}$ of constants of sort $\tau$
such that (the restriction of) $\xJ$ to $\Val_\tau$ forms a bijection $\Val_\tau
\cong \xI(\tau)$. We let $\Val = \bigcup_{\tau \in \xSTh} \Val_\tau$, whose
elements are called \emph{values}. For simplicity, we do not distinguish between
$c \in \Val$ and $\xJ(c)$. Note that, in~\cite{KN13frocos,FKN17tocl}, an
arbitrary overlap between term and theory symbols is allowed provided it is
covered by values. For simplicity, we assume $\xFTh \cap \xFTe = \varnothing$.

A \emph{valuation} over a model $\xM = \langle \xI, \xJ \rangle$ is a family
$\rho = (\rho_\tau)_{\tau \in \xSTh}$ of mappings $\rho_\tau\colon \xV^\tau \to
\xI(\tau)$. The \emph{interpretation} $\inter{t}_{\xM,\rho} \in \xI(\tau)$ of a
term $t^\tau \in \xT(\xFTh,\xV)$ in the model $\xM$ with respect to the
valuation $\rho = (\rho_\tau)_{\tau \in \xSTh}$ is inductively defined as
follows: $\inter{x^\tau}_{\xM,\rho} = \rho^\tau(x)$ and
$\inter{f(t_1,\ldots,t_n)}_{\xM,\rho} =
\xJ(f)(\inter{t_1}_{\xM,\rho},\ldots,\inter{t_n}_{\xM,\rho})$. We abbreviate
$\inter{t}_{\xM,\rho}$ as $\inter{t}_{\rho}$ if $\xM$ is known from the context.
Furthermore, for any ground term $t \in \xT(\xFTh)$, the valuation $\rho$ has no
impact on the interpretation $\inter{t}_\rho$ which can be safely ignored and
written as $\inter{t}$.

We suppose a special sort $\Bool \in \xSTh$ such that $\xI(\Bool) = \mathbb{B} =
\SET{\mathsf{true}, \mathsf{false}}$, and usual logical connectives 
$\neg, {\land},{\lor},\ldots \in \xFTh$ with their default sorts. We assume that
there exists for each  $\tau \in \xSTh$ an equality symbol $=_\tau$ of sort
$\tau \times \tau \to \Bool$ in $\xFTh$. For brevity we will omit ${}_\tau$
from $=_\tau$. We assume, for all of these theory symbols, that their
interpretation functions model their default semantics.
The terms in $\xT(\xFTh,\xV)^\Bool$ are called \emph{logical constraints}.%
\footnote{Logical constraints are quantifier-free, which is not restrictive:
Consider, for example, a formula $\forall x.\ \varphi$ with $n$ free variables
$x_1,\ldots,x_n$ and a quantifier-free formula $\varphi$. By introducing an
$n$-ary predicate symbol $p$ defined as $\inter{p(x_1,\ldots,x_n)}_{\xM,\rho} =
\inter{\forall x.\ \varphi}_{\xM,\rho}$, we can replace the formula by the
quantifier-free formula $p(x_1,\ldots,x_n)$. Clearly, this applies to arbitrary
first-order formulas. Another approach can be seen in~\cite[Section~2.2]{FKN17tocl}.} Note
that $\Var(\varphi) \subseteq \xVTh$ for any logical constraint $\varphi$, thus
in this case $\xT(\xFTh,\xV)^\Bool =\xT(\xFTh,\xVTh)^\Bool$. We say that a
logical constraint $\varphi$ is over a set $X \subseteq \xVTh$ of theory
variables if $\xV(\varphi) \subseteq X$. A logical constraint $\varphi$ is said
to be \emph{valid} in a model $\xM$, written as $\models_{\xM} \varphi$ (or
$\models \varphi$ when the model $\xM$ is known from the context), if
$\inter{\varphi}_{\xM,\rho} = \mathsf{true}$ for any valuation $\rho$ over the
model $\xM$. Considering the bijection $\Val_\tau  \cong \xI(\tau)$, an
arbitrary substitution $\sigma$ is equivalent to a valuation $\rho$. Suppose that
$\VDom(\sigma) = \SET{ x \in \Dom(\sigma) \mid \sigma(x) \in \Val }$ and
$\Var(\varphi) \subseteq \VDom(\sigma)$. Then the substitution $\sigma$ can be
seen as a valuation over $\varphi$, and $\models_{\xM} \varphi \sigma$
coincides with $\inter{\varphi}_{\xM,\sigma} = \mathsf{true}$. More
generally, we have the following.

\begin{restatable}{lemma}{LemmaSubstitutionAsInterpretation} 
\label{lem:substitution as interpretation}
Let $t \in \xT(\xFTh,\xVTh)$, $\rho$ a valuation, and $\sigma$ a substitution.
\begin{enumerate}
    \item  Suppose $\sigma(x) \in \xT(\xFTh,\xVTh)$ for all $x \in \xVTh$. 
    Let $\inter{\sigma}_{\xM,\rho}$ be a valuation defined as
    $\inter{\sigma}_{\xM,\rho}(x) = \inter{\sigma(x)}_{\xM,\rho}$. Then,
    $\inter{t}_{\xM,\inter{\sigma}_{\xM,\rho}} = \inter{t\sigma}_{\xM,\rho}$.
    \item Suppose that $\Var(t) \subseteq \VDom(\sigma)$. Then,
    $\inter{t}_{\xM,\hat \sigma} = \inter{t\sigma}_{\xM}$, where the valuation
    $\hat \sigma$ is defined by $\hat \sigma(x^\tau) = \xi(\sigma(x)) \in
    \xI(\tau)$ for $x \in \VDom(\sigma)$, where $\xi$ is a bijection $\Val^\tau
    \cong \xI(\tau)$.
\end{enumerate}
\end{restatable}

\begin{proof}[Proof (Sketch)]
\Bfnum{1.} Use structural induction on $t \in \xT(\xFTh,\xVTh)$.
 \Bfnum{2.} Similar to \Bfnum{1}, using the assumption $\Var(t) \subseteq \VDom(\sigma)$.
\end{proof}

\newcommand{\ProofLemmaSubstitutionAsInterpretation}{
\ifthenelse{\boolean{OmitProofs}}{\LemmaSubstitutionAsInterpretation*}{}
\begin{proof}
\Bfnum{1.}
The proof proceeds by structural induction on $t \in \xT(\xFTh,\xVTh)$.
Suppose $t = x \in \xVTh$.
Then, 
$\inter{x}_{\inter{\sigma}_{\rho}}
= \inter{\sigma}_{\rho}(x)
= \inter{\sigma(x)}_{\rho}$ by definition.
Suppose $t = f(t_1,\ldots,t_n)$.
Then, using the induction hypothesis,
we have $\inter{t}_{\inter{\sigma}_{\rho}}
= \inter{f(t_1,\ldots,t_n)}_{\inter{\sigma}_{\rho}}
= \xJ(f)(\inter{t_1}_{\inter{\sigma}_{\rho}},\ldots,\inter{t_n}_{\inter{\sigma}_{\rho}})
$ $= \xJ(f)(\inter{t_1\sigma}_{\rho},\ldots,\inter{t_n\sigma}_{\rho})
= \inter{f(t_1\sigma,\ldots,t_n\sigma)}_{\rho}
= \inter{t\sigma}_{\rho}$.
\Bfnum{2.}
Take an arbitrary valuation $\rho$.
Then, as $\Var(t) \subseteq \VDom(\sigma)$,
we have $\inter{\sigma}_{\xM,\rho}(x) 
= \inter{\sigma(x)}_{\xM,\rho}
= \inter{\sigma(x)}_{\xM}
= \sigma(x)
= \hat \sigma(x)$ for $x \in \Var(t)$.
Thus, similar to \Bfnum{1}, it follows by induction on $t$
that 
$\inter{t}_{\xM,\hat \sigma} 
= \inter{t\sigma}_{\xM,\rho}$.
Since $\Var(t) \subseteq \VDom(\sigma)$,
we have $t\sigma \in \xT(\xFTh)$.
Thus 
$\inter{t}_{\xM,\hat \sigma} = \inter{t\sigma}_{\xM}$.
\end{proof}}
\ifthenelse{\boolean{OmitProofs}}{}{\ProofLemmaSubstitutionAsInterpretation}

From Lemma~\ref{lem:substitution as interpretation} the following
characterizations, which are used later on, are obtained. Note that $\models
\varphi = \m{true}$ ($\models \varphi = \m{false}$) if and only if $\models
\varphi$ ($\models \neg\varphi$), for a logical constraint $\varphi$.

\begin{restatable}{lemma}{NoRespSubstImplyFalse}
\label{lem:no-resp-subst-imply-false}
Let $\varphi$ be a logical constraint.
\begin{enumerate}
    \item $\models_{\xM} \varphi$ if and only if $\models_{\xM} \varphi\sigma$
    for all substitutions $\sigma$ such that $\Var(\varphi) \subseteq
    \VDom(\sigma)$.

    \item If $\models_{\xM} \varphi$, then $\models_{\xM} \varphi\sigma$ for all
    substitutions $\sigma$ such that $\sigma(x) \in \xT(\xFTh,\xVTh)$ for all $x
    \in \Var(\varphi) \cap \Dom(\sigma)$. 

\item
The following statements are equivalent: 
\textit{(1)} $\models_{\xM} \neg\varphi$,
\textit{(2)} $\not\models_{\xM} \varphi\sigma$ for all substitutions $\sigma$
such that $\Var(\varphi) \subseteq \VDom(\sigma)$, and
\textit{(3)} $\sigma \models_{\xM} \varphi$ for no substitution $\sigma$.

\end{enumerate}
Here, $\sigma \models_\xM \varphi$ denotes that
$\Var(\varphi) \subseteq \VDom(\sigma)$ and 
$\models_{\xM} \varphi \sigma$ hold.
\end{restatable}

\begin{proof}[Proof (Sketch)]
\Bfnum{1.}
($\Rightarrow$)
Let $\sigma$ be a substitution such that $\Var(\varphi) \subseteq \VDom(\sigma)$,
and $\hat\sigma$ be defined as in Lemma~\ref{lem:substitution as interpretation}.
Then, $\inter{\varphi}_{\xM,\hat\sigma} = \mathsf{true}$,
and hence $\inter{\varphi\sigma}_{\xM} = \mathsf{true}$
by Lemma~\ref{lem:substitution as interpretation}.
Therefore, $\models_{\xM} \varphi\sigma$.
($\Leftarrow$)
Let $\rho$ be a valuation over a model $\xM = \langle \xI, \xJ \rangle$.
Then, in the view of $\Val_\tau  \cong \xI(\tau)$,
we can take a substitution $\check \rho$
given by $\check \rho(x) = \rho(x) \in \Val$ for all $x \in \Var(\varphi)$.
Then, use Lemma~\ref{lem:substitution as interpretation}
to obtain 
 $\inter{\varphi}_{\xM,\rho} = \inter{\varphi}_{\xM,\hat {\check\rho}} 
 = \mathsf{true}$, from which $\models_{\xM} \varphi$ follows.
\Bfnum{2.}
Take a substitution $\sigma'$ such that 
$\sigma'(x) = \sigma(x)$ for $x \in \Var(\varphi)$
and $\sigma'(x) = x$ otherwise.
Then, using Lemma~\ref{lem:substitution as interpretation}, we have
$\inter{\varphi}_{\inter{\sigma}_\rho} 
= \inter{\varphi}_{\inter{\sigma'}_\rho} 
= \inter{\varphi\sigma'}_{\rho}
= \inter{\varphi\sigma}_{\rho}$.
Thus, $\inter{\varphi\sigma}_{\rho} = \mathsf{true}$ for any $\rho$.
Therefore, $\models_{\xM} \varphi\sigma$.
\Bfnum{3.}
Use \Bfnum{1}.
\end{proof}

\newcommand{\ProofNoRespSubstImplyFalse}{
\ifthenelse{\boolean{OmitProofs}}{\NoRespSubstImplyFalse*}{}
\begin{proof}\ 
\begin{enumerate}
\item 
($\Rightarrow$)
Suppose $\models_{\xM} \varphi$.
Then $\inter{\varphi}_{\xM,\rho} = \mathsf{true}$ for any valuation $\rho$ over $\xM$.
Thus, 
for any substitution $\sigma$ such that $\Var(\varphi) \subseteq \VDom(\sigma)$,
$\inter{\varphi}_{\xM,\hat\sigma} = \mathsf{true}$,
where $\hat\sigma$ is defined as in Lemma~\ref{lem:substitution as interpretation}.
Hence, by Lemma~\ref{lem:substitution as interpretation},
$\inter{\varphi\sigma}_{\xM} = \mathsf{true}$
for any substitution $\sigma$ such that $\Var(\varphi) \subseteq \VDom(\sigma)$.
Therefore, 
$\models_{\xM} \varphi\sigma$ holds
for all substitution $\sigma$
such that $\Var(\varphi) \subseteq \VDom(\sigma)$.
($\Leftarrow$)
Suppose 
$\models_{\xM} \varphi\sigma$ for all substitutions $\sigma$
such that $\Var(\varphi) \subseteq \VDom(\sigma)$.
Let $\rho$ be a valuation over a model $\xM = \langle \xI, \xJ \rangle$.
Then $\rho = (\rho_\tau)_{\tau \in \xSTh}$
for mappings $\rho_\tau\colon \xV^\tau \to \xI(\tau)$.
Now, in the view of $\Val_\tau  \cong \xI(\tau)$,
we can take a substitution $\check \rho$
given by $\check \rho(x) = \rho(x) \in \Val$ for all $x \in \Var(\varphi)$.
Then, as $\Var(\varphi) \subseteq \VDom(\check \rho)$,
we have $\models_{\xM} \varphi\check \rho$ by our assumption.
This means $\inter{\varphi\check \rho}_\xM = \mathsf{true}$.
Hence, it follows from Lemma~\ref{lem:substitution as interpretation} that
$\inter{\varphi}_{\xM,\hat {\check\rho}} = \mathsf{true}$. Moreover, we have
$\hat {\check\rho}(x) = \rho(x)$ for $x \in \Var(\varphi)$ by definition, and
thus $\inter{\varphi}_{\xM,\rho} = \mathsf{true}$. Therefore,
$\inter{\varphi}_{\xM,\rho} = \mathsf{true}$ for any valuation $\rho$ over
$\xM$. This shows $\models_{\xM} \varphi$.

\item
Suppose $\models_{\xM} \varphi$.
Then $\inter{\varphi}_{\xM,\rho} = \mathsf{true}$ for any valuation $\rho$ over $\xM$.
Suppose that $\sigma$ is a substitution such that
$\Var(\varphi) \cap \Dom(\sigma) \subseteq \xT(\xFTh,\xVTh)$.
Let $\sigma'$ be a substitution such that 
$\sigma'(x) = \sigma(x)$ for $x \in \Var(\varphi)$
and $\sigma'(x) = x$ otherwise.
Then $\sigma'(x) \in \xT(\xFTh,\xVTh)$ for any $x \in \xVTh$.
Then, using Lemma~\ref{lem:substitution as interpretation}, we have
$\inter{\varphi}_{\inter{\sigma}_\rho} 
= \inter{\varphi}_{\inter{\sigma'}_\rho} 
= \inter{\varphi\sigma'}_{\rho}
= \inter{\varphi\sigma}_{\rho}$.
Thus, $\inter{\varphi\sigma}_{\rho} = \mathsf{true}$ for any $\rho$.
Therefore, $\models_{\xM} \varphi\sigma$.

\item
Using \Bfnum{1}, we have 
$\models_{\xM} \neg \varphi$
iff 
$\models_{\xM} \neg \varphi\sigma$ for all substitution $\sigma$
such that $\Var(\varphi) \subseteq \VDom(\sigma)$
iff 
$\not \models_{\xM} \varphi\sigma$ for all substitution $\sigma$
such that $\Var(\varphi) \subseteq \VDom(\sigma)$
iff 
there is no substitution $\sigma$
such that $\models_{\xM} \varphi\sigma$ and $\Var(\varphi) \subseteq \VDom(\sigma)$
iff 
$\sigma \models \varphi$ for no substitution $\sigma$.
\end{enumerate}
\end{proof}
}
\ifthenelse{\boolean{OmitProofs}}{}{\ProofNoRespSubstImplyFalse}

LCTRSs admit special rewrite steps over $\xT(\xF,\xV)$
specified by the underlying model $\xM = \langle \xI, \xJ \rangle$.
Such rewrite steps are called \emph{calculation steps} and
denoted by $s \Rca t$, which is defined as follows:
$s \Rca t$ if $s = C[f(c_1,\ldots,c_n)]$ and $t = C[c_0]$
for $f \in \xFTh \setminus \Val$ and $c_0,\ldots,c_n \in \Val$ 
with $c_0 = \xJ(f)(c_1,\ldots,c_n)$ and a context $C$.
The following lemma connects calculation steps and interpretations over ground theory terms $\xT(\xFTh)$.
In the following $s \Ra[!] t$ is used for $s \Ra[*] t$ with $t$ being a normal form with respect to $\Ra$.

\begin{restatable}{lemma}{LemmaPropertiesOfCalculationSteps}
\label{lem:properties of calculation steps}
Let $s,t \in \xT(\xFTh)$. Then, all of the following holds:
\textup{\Bfnum{1.}}
$\inter{t} \in \Val$,
\textup{\Bfnum{2.}}
$t \Rca^{!} \inter{t}$,
\textup{\Bfnum{3.}}
$s \Rca^{*} t$ implies $\inter{s} = \inter{t}$, and 
\textup{\Bfnum{4.}}
$s \Cca^{*} t$ if and only if $\inter{s} = \inter{t}$.
\end{restatable}

\begin{proof}[Proof (Sketch)]
\Bfnum{1.}
This claim follows as $\inter{t^\tau} \in \xI(\tau) \cong \Val^\tau$.
\Bfnum{2.}
Show $t \Rca^{*} \inter{t}$ by structural induction on $t$.
Then, the claim follows, 
since values are normal forms with respect to calculation steps.
\Bfnum{3.}
We use the fact that 
the set of calculation rules forms a confluent LCTRS~\cite{KN13frocos}.
Since $s \Rca^{!} \inter{s}$ and $t \Rca^{!} \inter{t}$ from \Bfnum{2},
$s \Rca^{*} t$ implies $\inter{s} = \inter{t}$ by confluence.
\Bfnum{4.}
The \textit{only-if} part follows from \Bfnum{3},
and the \textit{if} part follows from \Bfnum{1}.
\end{proof}

\newcommand{\ProofLemmaPropertiesOfCalculationSteps}{
\ifthenelse{\boolean{OmitProofs}}{\LemmaPropertiesOfCalculationSteps*}{}
\begin{proof}
\Bfnum{1.}
This claim follows as $\inter{t^\tau} \in \xI(\tau) \cong \Val^\tau$.
\Bfnum{2.}
To show this,
it suffices to show $t \Rca^{*} \inter{t}$,
since values are normal forms with respect to calculation steps.
We show $t \Rca^{*} \inter{t}$ by structural induction on $t$.
Suppose $t \in \Val$, then
$\inter{t} = t$ and hence $t \Rca^{*} \inter{t}$.
Suppose $t \not\in \Val$.
Then $t = f(t_1,\ldots,t_n)$ for some 
$f \in \xFTh \setminus \Val$ and $t_1,\ldots,t_n \in \xT(\xFTh)$.
By the induction hypothesis, we have $t_i \Rca^{*} \inter{t_i}$ for all $1 \le i \le n$.
Furthermore, by the definition of the calculation steps
we have $f(\inter{t_1},\ldots,\inter{t_n}) \to \xI(f)(\inter{t_1},\ldots,\inter{t_n})
= \inter{f(t_1,\ldots,t_n)}$.
Therefore, $t = f(t_1,\ldots,t_n) \Rca^{*} f(\inter{t_1},\ldots,\inter{t_n}) \Rca \inter{t}$.
\Bfnum{3.}
First, note that the set of calculation rules
$\xRca = \SET{ f(x_1,\ldots,x_n) \to x_0 ~ \CO{x_0=f(x_1,\ldots,x_n)
\mid f\colon \sig{\tau} \to \tau_0 \in \xFTh \setminus \Val, $ $
x_i \in \Var^{\tau_i} \text{ for all } 0 \le i \le n }$~\cite{KN13frocos}}
is terminating and has no overlaps, and hence is confluent.
Suppose $s \Rca^{*} t$.
Then it follows from \Bfnum{2} that $s \Rca^{!} \inter{s}$ and
$t \Rca^{!} \inter{t}$. Hence, $\inter{s} = \inter{t}$ by confluence.
This shows the desired claim.
\Bfnum{4.}
The \textit{only-if} part of this claim follows immediately from
\Bfnum{3}. The \textit{if} part follows from \Bfnum{1} as $s \Rca^{*} \inter{s}
= \inter{t} \Lca^{*} t$.
\end{proof}
}
\ifthenelse{\boolean{OmitProofs}}{}{\ProofLemmaPropertiesOfCalculationSteps}

The other type of rewrite steps in LCTRSs are rule steps specified by rewrite
rules. Let us fix a signature $\langle \xSTh,\xSTe,\xFTh,\xFTe \rangle$. A
constrained rule of an LCTRS is a triple $\ell \R r~\CO{\varphi}$ of terms
$\ell,r$ with the same sort satisfying $root(\ell) \in \xFTe$ and a logical
constraint $\varphi$. We define $\LVar(\ell \R r~\CO{\varphi}) = (\Var(r)
\setminus \Var(\ell))\cup \Var(\varphi)$, whose members are called \emph{logical
variables} of the rule. The intention is that the logical variables of rules in
LCTRSs are required to be instantiated only by values. Let us also fix a model
$\xM$. Then, a substitution $\gamma$ is said to \emph{respect a rewrite
rule} $\ell \R r~\CO{\varphi}$ if $\LVar(\ell \R r~\CO{\varphi}) \subseteq
\VDom(\gamma)$ and $\models_{\xM} \varphi\gamma$. Using this notation, a rule
step $s \Rru t$ over the model $\xM$ by the rewrite rule $\ell \R
r~\CO{\varphi}$ is given as follows: $s \Rru t$ if and only if $s =
C[\ell\gamma]$ and $t = C[r\gamma]$ for some context $C$ and some substitution
$\gamma$ that respects the rewrite rule $\ell \R r~\CO{\varphi}$.

Finally, a \emph{logically constrained term rewrite system} (LCTRS, for short)
consists of a signature $\Sigma = \langle \xSTh,\xSTe,\xFTh,\xFTe \rangle$, a
model $\xM$ over $\SigmaTh = \langle \xSTh, \xFTh \rangle$ (which induces the
set $\Val \subseteq \xFTh$ of values) and a set $\xR$ of constrained rules over
the signature $\Sigma$. All this together defines rewrite steps consisting of
calculation steps and rule steps. In a practical setting, often some predefined
(semi-)decidable theories are assumed and used as model $\xM$ and theory
signature $\langle \xSTh, \xFTh \rangle$. An example of such a theory is
\emph{linear integer arithmetic}, whose model consists of standard boolean
functions and the set of integers including standard predefined functions on
them. From this point of view, we call the triple $\fU = \langle \xSTh, \xFTh,
\xM \rangle$ of the theory signature and its respective model the
\emph{underlying model} or \emph{background theory} of the LCTRS. We also denote
an LCTRS as $\langle \xM, \xR \rangle$ with an implicit signature or $\langle
\xM, \xR \rangle$ over the signature $\Sigma = \langle \xSTh, \xSTe, \xFTh,
\xFTe \rangle$ for an explicit signature.

\section{Validity of Constrained Equational Theories}
\label{sec:constrained-equational-validity}

In this section, we introduce validity of constrained equational theories
(CE-validity), which is a key concept used throughout the paper. Subsequently,
we present fundamental properties of CE-validity, and show their relation to the
conversion of rewriting.

\subsection{Constrained Equational Theory and Its Validity}

In this subsection, after introducing the notion of CEs,
we define equational systems, which are sets of CEs, and rewriting with respect
to such systems. This gives an equational version of the rewrite step in LCTRSs.
Furthermore, based on these notions, we define the validity of CEs.

Recall that logical variables of a constrained rule are those which are only
allowed to be instantiated by values. As we have seen in the previous section,
rewrite steps of LCTRSs depend on the correct instantiation of the logical
variables of the applied rule. However, the sets of logical variables
$\LVar(\ell \R r~\CO{\varphi})$ and $\LVar(r \R \ell~\CO{\varphi})$ are not
necessarily equivalent, and the CE $\ell \approx r~\CO{\varphi}$ alone does not
suffice to specify the correct logical variables. This motivates us to add an
explicit set $X$ to the CE $\ell \approx r~\CO{\varphi}$ as
$\CEqn{X}{\ell}{r}{\varphi}$ which specifies its logical variables.

\begin{definition}[constrained equation]
Let $\SigmaTe = \langle \xSTe, \xFTe \rangle$ be a term signature over the
underlying model $\fU = \langle \xSTh, \xFTh,\xM \rangle$. A \emph{constrained
equation} (CE, for short) over $\fU$ and $\SigmaTe$ is a quadruple
$\CEqn{X}{s}{t}{\varphi}$ where $s,t$ are terms with the same sort, $\varphi$ is
a logical constraint, and $X \subseteq \xVTh$ is a set of theory variables
satisfying $\Var(\varphi) \subseteq X$. A \emph{logically constrained equational
system} (LCES, for short) is a set of CEs. We abbreviate
$\CEqn{X}{s}{t}{\varphi}$ to $\CEqn{}{s}{t}{\varphi}$ if $\Var(\varphi) = X$. A
CE $\CEqn{X}{s}{t}{\mathsf{true}}$ is abbreviated to
$\CEqn{X}{s}{t}{}$.
\end{definition}
We remark that a constrained rewrite rule $\ell \to r ~\CO{\varphi}$ is
naturally encoded as a CE $\CEqn{X}{\ell}{r}{\varphi}$ by taking  $X =
\LVar(\ell \to r~\CO{\varphi})$. Furthermore, let us illustrate the
aforementioned issues, without an explicit set of logical variables, by an
example.

\begin{example}
Consider the LCTRS $\xR$ over the theory of integer arithmetic and its (labeled) rules 
\begin{align*}
    \alpha\colon \m{f}(x,y) &\R \m{g}(z)~\CO{x = \m{1}} 
    & \beta\colon \m{g}(z) &\R \m{f}(x,y)~\CO{x = \m{1}}
\end{align*}
with their sets of logical variables $\LVar(\alpha) = \SET{x,z}$ and
$\LVar(\beta) = \SET{x,y}$. Transforming them naively into the CE $\m{f}(x,y)
\approx \m{g}(z)~\CO{x = \m{1}}$ and $\m{g}(z) \approx \m{f}(x,y)~\CO{x =
\m{1}}$ would give the set of logical variables $\SET{x}$ for both. We use the
notion of logical variables in Winkler and Middeldorp~\cite{WM18}, where the set
of logical variables of a CE consists of the variables appearing in the
constraint. Obviously, we lose concrete information about logical variables of
the original rules. Clearly, in our notion this information remains intact:
$\CEqn{\SET{x,z}}{\m{f}(x,y)}{\m{g}(z)}{x = \m{1}}$ and
$\CEqn{\SET{x,y}}{\m{g}(z)}{\m{f}(x,y)}{x = \m{1}}$. Note that variables
appearing solely in the set of logical variables and not in the CE have no
effect but are allowed. For example, in the CE
$\CEqn{\SET{x,z,z'}}{\m{f}(x,y)}{\m{g}(z)}{x = \m{1}}$ the logical variable $z'$
has no effect and could be dropped. 
\end{example}

In the following we extend the notion of rewrite steps by using CEs
instead of rewrite rules.

\begin{definition}[{$\Cab[\xE]{}$}]
\label{def:equational-rewriting}
Let $\xE$ be an LCES over the underlying model $\fU = \langle \xSTh, \xFTh,\xM
\rangle$ and the term signature $\SigmaTe = \langle \xSTe, \xFTe \rangle$. For
terms $s,t \in \xT(\xF,\xV)$, we define a \emph{rule} step $s \Cru t $ if $s =
C[\ell\sigma]$ and $t = C[r\sigma]$ (or vice versa) for some CE
$\CEqn{X}{\ell}{r}{\varphi} \in \xE$ and some $X$-valued substitution
$\sigma$ such that $\models_{\xM} \varphi\sigma $. Here, a substitution is said
to be \emph{$X$-valued} if $X \subseteq \VDom(\sigma)$. We let ${\Cab[\xE]{}} =
{\Cca} \cup {\Cru}$, where ${\Cca}$ is the symmetric closure of the calculation
steps ${\Rca}$ specified by $\xM$.
\end{definition}

We give examples on rewriting with CEs.

\begin{example}
\label{exa:modular-arithmetic-es}
Consider integer arithmetic as underlying model $\xM$. We consider the term
sorts $\xSTe = \SET{ \m{Unit} }$ and the term signature $\xFTe =
\SET{\m{cong}\colon \m{Int} \to \m{Unit}}$ where $\m{Int}$ is the
respective sort of the integers. The set $\xE$ of CEs consists of
$\SET{\m{cong}(x) \approx \m{cong}(y)~\CO{\m{mod}(x,12) = \m{mod}(y,12)}}$.
Arithmetic values in intermediate steps of rewrite sequences wrapped in
$\m{cong}$ have the property that they are \emph{congruent modulo 12} and thus
$\xE$ simulates modular arithmetic with modulus 12. Consider the following
sequence: 
\begin{align*} 
\m{cong}(7 + 31) \Cca \m{cong}(38) \Cru[\xE]{} \m{cong}(14) 
\end{align*}
From this we conclude that $7 + 31$, which gives $38$, and $14$ are congruent
modulo 12. Note that the rule step $\Cru[\xE]{}$ does not allow to directly
convert $\m{cong}(7 + 31)$ and $\m{cong}(14)$.
\end{example}

\begin{example}
\label{exa:group with exponentials}
Consider integer arithmetic as the underlying model $\xM$.
We take a term signature
$\xSTe = \SET{ \m{G} }$
and 
$\xFTe = \SET{
\m{e}\colon \m{G},
\m{inv}\colon \m{G} \to \m{G},
\m{*}\colon \m{G} \times \m{G} \to \m{G},
\m{exp}\colon \m{G} \times \m{Int} \to \m{G}}$.
Let the set $\xE$ of CEs consist of:
\[
\begin{array}{rcll@{\qquad}rcll}
(x * y) * z      &\approx&  x * (y * z)& &
  \m{e} * x      &\approx&  x    & \\
\m{inv}(x) * x   &\approx& \m{e} & &
\m{exp}(x,\m{0}) &\approx& \m{e} & \\
\m{exp}(x,\m{1}) &\approx& x     & &
\Pi \SET{ n,m }.\, \m{exp}(x,n) * \m{exp}(x,m) &\approx& \m{exp}(x,m+n) &
\end{array}
\]
As in first-order equational reasoning,
one can show $x *  \m{e} \Cab[\xE]{*} x$.
Thus,
$\m{exp}(x,\m{-1}) 
\Cab[\xE]{}  \m{e} * \m{exp}(x,\m{-1})
\Cab[\xE]{}  (\m{inv}(x) * x) * \m{exp}(x, \m{-1})
\Cab[\xE]{}  \m{inv}(x) * (x* \m{exp}(x,\m{-1}))
\Cab[\xE]{}  \m{inv}(x)*(\m{exp}(x,\m{1}) * \m{exp}(x, \m{-1}))
\Cab[\xE]{}  \m{inv}(x)* \m{exp}(x,\m{1} + (\m{-1}))
\Cab[\xE]{}   \m{inv}(x) * \m{exp}(x, \m{0})
\Cab[\xE]{}   \m{inv}(x)*\m{e}
\Cab[\xE]{*}  \m{inv}(x)
$
as expected. This encodes a system of groups with an explicit 
exponentiation operator $\m{exp}$.
\end{example}

\begin{example}
\label{exp:list functions}
Consider integer arithmetic as the underlying model $\xM$.
We take a term signature
$\xSTe = \SET{\m{Elem}, \m{List},\m{ElemOp}}$
and
$\xFTe = \SET{%
\m{nil}\colon \m{List},\allowbreak
\m{cons}\colon \m{Elem} \times \m{List} \to \m{List},\allowbreak
\m{none}\colon \m{ElemOp},\allowbreak
\m{some}\colon \m{Elem} \to \m{ElemOp},\allowbreak
\m{length}\colon \m{List}  \to \m{Int},\allowbreak
\m{nth}\colon \m{List} \times \m{Int} \to \m{ElemOp}\allowbreak
}$.
Let the set $\xE$ of CEs consist of
\[
\begin{array}{rcll@{\qquad}rcll}
\m{length}(\m{nil}) &\approx& \m{0} & &
\m{length}(\m{cons}(x,xs)) &\approx& \m{length}(xs) + \m{1}\\
\Pi \SET{ n }.\, \m{nth}(\m{nil},n) &\approx& \m{none} & &
\m{nth}(xs,n) &\approx& \m{none} & \CO{n < \m{0}}\\
\m{nth}(\m{cons}(x,xs),\m{0}) &\approx& \m{some}(x) & &
\m{nth}(\m{cons}(x,xs),n) &\approx& \m{nth}(xs,n-\m{1}) & \CO{n > \m{0}}\\
\end{array}
\]
This LCES encodes common list functions that use integers. For program
verification purposes, one may deal with the validity problem of a formula such
as $\m{nth} (xs,n) \not\approx \m{none} \Leftrightarrow \m{0} \leqslant n
\land n < \m{length}(xs)$.
\end{example}

We continue by giving some immediate facts which are used later on.

\begin{restatable}{lemma}{LemmaSymAndClosurePropertiesOfEquationalSteps}   
\label{lem:sym and closure properties of equational steps}
Let $\xE$ be an LCES over the underlying model $\fU = \langle \xSTh,
\xFTh,\xM \rangle$ and the term signature $\SigmaTe = \langle \xSTe, \xFTe
\rangle$. Then, all of the following hold:
\textup{\Bfnum{1.}}
${\Cab[\xE]{}}$ is symmetric,
\textup{\Bfnum{2.}}
${\Cab[\xE]{}}$ is closed under contexts
i.e.\ $s \Cab[\xE]{} t$ implies $C[s] \Cab[\xE]{} C[t]$ for any context $C$, and 
\textup{\Bfnum{3.}}
${\Cab[\xE]{}}$ is closed under substitutions,
i.e.\ $s \Cab[\xE]{} t$ implies $s\sigma \Cab[\xE]{} t\sigma$ 
for any substitution $\sigma$.
\end{restatable}

\begin{proof}[Proof (Sketch)]
\Bfnum{1} and~\Bfnum{2} are trivial. For~\Bfnum{3} the case $s \Cca t$ is clear.
Suppose $s \Cru t$.
Then $s = C[\ell\rho]$ and $t = C[r\rho]$ (or vice versa)
for some CE $\CEqn{X}{\ell}{r}{\varphi} \in \xE$
and an $X$-valued substitution $\rho$ such that $\models_{\xM} \varphi\rho$.
Then $\varphi\rho = (\varphi\rho)\sigma = \varphi(\sigma \circ
\rho)$ and hence $\models_{\xM} \varphi(\sigma \circ \rho)$. 
Then, the claim follows, as 
$s\sigma = C[\ell\rho]\sigma = C\sigma[\ell(\sigma \circ \rho)]$ and $t\sigma =
C[r\rho]\sigma = C\sigma[r(\sigma \circ \rho)]$.
\end{proof}

\newcommand{\ProofSymAndClosurePropertiesOfEquationalSteps}{
\ifthenelse{\boolean{OmitProofs}}{\LemmaSymAndClosurePropertiesOfEquationalSteps*}{}
\begin{proof}
\Bfnum{1} and~\Bfnum{2} are trivial. For~\Bfnum{3} the case $s \Cca t$ is clear.
So, we consider $s \Cru t$ and have a CE $\CEqn{X}{\ell}{r}{\varphi} \in \xE$,
an $X$-valued substitution $\rho$ such that $\models_{\xM} \varphi\rho$, a
context $C$ such that $s = C[\ell\rho]$ and $t = C[r\rho]$ (or vice versa).
Then, for any $x \in X$, $\sigma(\rho(x)) = \rho(x)$ as $\rho(x) \in \Val$.
Thus, $\sigma \circ \rho$ is again an $X$-valued substitution. Moreover, as
$\Var(\varphi) \subseteq X \subseteq \VDom(\rho)$, we have $\varphi\rho \in
\xT(\xFTh)$. Also, $\varphi\rho = (\varphi\rho)\sigma = \varphi(\sigma \circ
\rho)$ and hence $\models_{\xM} \varphi(\sigma \circ \rho)$. The substitution
$\sigma \circ \rho$ is an $X$-valued substitution with $\models_{\xM}
\varphi(\sigma \circ \rho)$. Therefore, we have $s\sigma \Cru t\sigma$, because
$s\sigma = C[\ell\rho]\sigma = C\sigma[\ell(\sigma \circ \rho)]$ and $t\sigma =
C[r\rho]\sigma = C\sigma[r(\sigma \circ \rho)]$.
\end{proof}
}
\ifthenelse{\boolean{OmitProofs}}{}{\ProofSymAndClosurePropertiesOfEquationalSteps}

We proceed by defining constrained equational theories (CE-theories) and
validity of CEs (CE-validity) with respect to a CE-theory.

\begin{definition}[constrained equational theory]
A \emph{constrained equational theory} is specified by a triple
$\fT = \langle \fU, \SigmaTe, \xE \rangle$, where $\fU = \langle\xSTh,\xFTh, \xM
\rangle$ is an underlying model, $\SigmaTe$ is a term signature over $\fU$ (as
given in the LCTRS formalism), and $\xE$ is an LCES over $\fU,
\SigmaTe$. If no confusion arises, we refer to the CE-theory by
$\langle \xM, \xE \rangle$, without stating its signature explicitly. We also
say that a CE-theory $\langle \xM, \xE \rangle$ is defined over the
signature $\Sigma = \langle \xSTh, \xSTe, \xFTh, \xFTe \rangle$ in order to make
the signature explicit.
\end{definition}

\begin{definition}[CE-validity]
Let $\fT = \langle\xM, \xE \rangle$ be a CE-theory. Then a CE
$\CEqn{X}{s}{t}{\varphi}$ is said to be a \emph{constrained equational
consequence} (CE-consequence, for short) \emph{of $\fT$}  or
\emph{valid} (CE-valid, for clarity), written as $\fT \cec
\CEqn{X}{s}{t}{\varphi}$, if $s\sigma \Cab[\xE]{*} t\sigma$ for all $X$-valued
substitutions $\sigma$ such that $\models_{\xM} \varphi\sigma$. We write $\xE
\cec \CEqn{X}{s}{t}{\varphi}$ if $\xM$ is known from the context.
\end{definition}    

We conclude this subsection with an example on CE-validity.

\begin{example}
\label{exp:abs max}
Consider integer arithmetic as the underlying model $\xM$.
We take the term signature
$\xFTe = \SET{ 
\m{abs}\colon \m{Int} \to \m{Int},
\m{max}\colon \m{Int} \times \m{Int} \to \m{Int} 
}$ 
the set of CEs $\xE$ consisting of
\[
\begin{array}{rcll@{\qquad\qquad}rcll}
\m{abs}(x) &\approx& - x & \CO{x < \m{0}} &
\m{abs}(x) &\approx& x   & \CO{x \ge \m{0}}\\
\m{max}(x,y) &\approx& x  & \CO{ x \ge y} &
\m{max}(x,y) &\approx& y   & \CO{x < y }\\
\end{array}
\]
The following are valid CE-consequences:
\[
\begin{array}{l@{\qquad\qquad}l}
    \fT \cec \CEqn{\SET{ x }}{\m{abs}(x)}{\m{abs}(-x)}{} &
    \fT \cec \CEqn{\SET{ x, y }}{\m{max}(x,y)}{\m{max}(y,x)}{} \\
\multicolumn{2}{l}{\fT \cec \CEqn{\SET{ x, y }}{\m{abs}(\m{max}(x,y))}{\m{max}(\m{abs}(x),\m{abs}(y))}{0 \le x \land 0 \le y }}
\end{array}
\]
On the other hand, the CE $\CEqn{\varnothing}{\m{abs}(x)}{\m{abs}(-x)}{}$
is not a valid CE-consequence: For the $\varnothing$-valued identity
substitution $\sigma$, we have that $\m{abs}(x)\sigma = \m{abs}(x)
\not\Cab[\xE]{*} \m{abs}(-x) = \m{abs}(-x)\sigma$.
\end{example}

\subsection{Properties of CE-Validity}

This subsection covers important properties related to 
CE-validity, for example, we show that 
validity forms an equivalence and a congruence relation.
Furthermore, we cover in which way it is closed under substitutions and contexts,
and how equality can be induced from constraints.

Our first two lemmas follow immediately from the definition of the CE-validity.

\begin{restatable}{lemma}{LemmaRule}
\label{lem:rule}
Let $\fT = \langle \xM, \xE \rangle$  be a CE-theory. Then for any
$\CEqn{X}{s}{t}{\varphi} \in \xE$, we have $\fT \cec \CEqn{X}{s}{t}{\varphi}$.
\end{restatable}

\newcommand{\ProofRule}{
\ifthenelse{\boolean{OmitProofs}}{\LemmaRule*}{}
\begin{proof}
Consider a CE-theory $\fT = \langle\xM, \xE \rangle$ and an arbitrary CE
$\CEqn{X}{s}{t}{\varphi} \in \xE$. Let $\sigma$ be an $X$-valued substitution
such that $\models_{\xM} \varphi\sigma$. Then, we have $s\sigma \Cab[\m{rule},
\xE]{} t\sigma$ by definition of rewriting with equations. Hence, $s\sigma
\Cab[\xE]{*} t\sigma$ and therefore $\fT \cec \CEqn{X}{s}{t}{\varphi}$.
\end{proof}
}
\ifthenelse{\boolean{OmitProofs}}{}{\ProofRule}

\begin{restatable}[congruence]{lemma}{LemmaCongruence}
\label{lem:congruence}
Let $\fT = \langle\xM, \xE \rangle$  be a CE-theory. For any set $X
\subseteq \xVTh$ and logical constraint $\varphi$ such that $\Var(\varphi)
\subseteq X$, the binary relation $\fT \cec \CEqn{X}{\cdot}{\cdot}{\varphi}$
over terms is a congruence relation over $\Sigma$.
\end{restatable}

\newcommand{\ProofCongruence}{
\ifthenelse{\boolean{OmitProofs}}{\LemmaCongruence*}{}
\begin{proof}
Consider a CE-theory $\fT$. Furthermore, we assume a set $X \subseteq \xVTh$ and
a logical constraint $\varphi$ such that $\Var(\varphi) \subseteq X$. In order
for the CE-consequence to be a binary congruence relation over terms, it needs
to be reflexive, symmetric, transitive and congruent. We prove those properties
independently:
\begin{enumerate}
\item Trivially $s\sigma \Cab[\xE]{*} s\sigma$ holds for all ($X$-valued)
substitutions $\sigma$. Hence also $\fT \cec \CEqn{X}{s}{s}{\varphi}$.
(reflexivity)
\item 
By Lemma~\ref{lem:sym and closure properties of equational steps},
we have $u \Cab[\xE]{} v$ if and only if $v \Cab[\xE]{} u$. By repeated application
we obtain $u \Cab[\xE]{*} v$ if and only if $v \Cab[\xE]{*} u$.
Consider that $\fT \cec \CEqn{X}{s}{t}{\varphi}$. Then we have for all
$X$-valued substitutions $\sigma$ with $\models_{\xM} \varphi\sigma$ that
$s\sigma \Cab[\xE]{*} t\sigma$, which further implies $t\sigma \Cab[\xE]{*}
s\sigma$. Therefore, $\fT \cec \CEqn{X}{t}{s}{\varphi}$. (symmetry)

\item Consider that $\fT \cec \CEqn{X}{s}{t}{\varphi}$ and $\fT \cec
\CEqn{X}{t}{u}{\varphi}$. This implies that for all $X$-valued substitutions
$\sigma$ with $\models_{\xM} \varphi\sigma$ that $s\sigma \Cab[\xE]{*} t\sigma
\Cab[\xE]{*} u\sigma$, which gives further that $s\sigma \Cab[\xE]{*} u\sigma$.
Therefore, $\fT \cec \CEqn{X}{s}{u}{\varphi}$ follows. (transitivity)

\item Consider that $\fT \cec \CEqn{X}{s_i}{t_i}{\varphi}$ for all $1 \leqslant
i \leqslant n$. We have for all $X$-valued substitutions $\sigma$ with
$\models_{\xM} \varphi\sigma$ that $s_i\sigma \Cab[\xE]{*} t_i\sigma$ for all $1
\leqslant i \leqslant n$, which implies that $f(\seq{s})\sigma =
f(s_1\sigma,\ldots,s_n\sigma) \Cab[\xE]{*} f(t_1\sigma,\ldots,t_n\sigma) =
f(\seq{t})\sigma$ for a $f \in \Sigma$ by Lemma~\ref{lem:sym and closure
properties of equational steps}. Therefore $\fT \cec
\CEqn{X}{f(\seq{s})}{f(\seq{t})}{\varphi}$. (congruence)
\end{enumerate}
All properties in order to form a congruence relation are fulfilled.
\end{proof}
}
\ifthenelse{\boolean{OmitProofs}}{}{\ProofCongruence}

For stability under substitutions, we differentiate two kinds;
for each CE $\CEqn{X}{s}{t}{\varphi}$,
the first one considers substitutions instantiating variables in $X$; the second
one considers substitutions instantiating variables not in $X$.

\begin{restatable}[stability of theory terms]{lemma}{LemmaStability}
\label{lem:stability}
Let $\fT = \langle \xM, \xE \rangle$ be a CE-theory.
Let $X,Y \subseteq \xVTh$ be sets of theory variables and $\sigma$ a
substitution such that $\sigma(y) \in \xT(\xFTh,X)$ for any $y \in Y$. If $\fT
\cec \CEqn{Y}{s}{t}{\varphi}$, then $\fT \cec
\CEqn{X}{s\sigma}{t\sigma}{\varphi\sigma}$.
\end{restatable}

\begin{proof}[Proof (Sketch)]
Take any $X$-valued substitution $\theta$ with $\models_{\xM}
(\varphi\sigma)\theta$. This gives a $Y$-valued substitution $\xi$ by defining
$\xi(y) = \inter{(\theta \circ \sigma)(y)}$ for each $y \in Y$. From
Lemma~\ref{lem:properties of calculation steps}, we know $(\theta \circ
\sigma)(y) \Cab[\xE]{*} \xi(y)$ for any $y \in Y$. We also have $\models_{\xM}
\varphi\xi$ by Lemma~\ref{lem:substitution as interpretation}, and hence $s\xi
\Cab[\xE]{*} t\xi$ by assumption. Thus, using Lemma~\ref{lem:sym and closure
properties of equational steps}, we obtain $(s\sigma)\theta = s(\theta \circ
\sigma) \Cab[\xE]{*} s\xi \Cab[\xE]{*} t\xi \Cab[\xE]{*} t(\theta \circ \sigma)
= (t\sigma)\theta$.
\end{proof}

\newcommand{\ProofStability}{
\ifthenelse{\boolean{OmitProofs}}{\LemmaStability*}{}
\begin{proof}
Consider a CE-theory $\fT = \langle\xM, \xE \rangle$,
the sets $X,Y \subseteq \xVTh$ and the substitution $\sigma$ such that
$\sigma(y) \in \xT(\xFTh, X)$ for any $y \in Y$. Moreover, assume 
$\fT \cec \CEqn{Y}{s}{t}{\varphi}$ which gives $s\gamma \Cab[\xE]{*} t\gamma$ for all
$Y$-valued substitutions $\gamma$ such that $\models_{\xM} \varphi\gamma$.
Let $\theta$ be an $X$-valued substitution such that 
$\models_{\xM} (\varphi\sigma)\theta$,
i.e.\ $\models_{\xM} \varphi(\theta \circ \sigma)$.
Take $y \in Y$. Then, we have $\sigma(y) \in \xT(\xFTh,X)$ by assumption,
and thus $\theta(\sigma(y)) \in \xT(\xFTh)$ as $\theta$ is $X$-valued.
Hence, we have $(\theta \circ \sigma)(y) \in \xT(\xFTh)$ for all $y \in Y$.
Now, take for each $y \in Y$,  
the value $c_y = \inter{(\theta \circ \sigma)(y)}$.
This gives a $Y$-valued substitution $\xi$ by defining 
$\xi(y) = c_y$ for each $y \in Y$.
From this we obtain $(\theta \circ \sigma)(y) \Cab[\xE]{*} \xi(y)$ for any $y \in Y$
by Lemma~\ref{lem:properties of calculation steps},
and thus
$s(\theta \circ \sigma) \Cab[\xE]{*} s\xi$ 
and
$t(\theta \circ \sigma) \Cab[\xE]{*} t\xi$ 
by Lemma~\ref{lem:sym and closure properties of equational steps}.
Furthermore, 
we have $\inter{\varphi(\theta \circ \sigma)} = \inter{\varphi\xi}$
by Lemma~\ref{lem:substitution as interpretation},
and thus, by $\models_{\xM} \varphi(\theta \circ \sigma)$,
we obtain $\models_{\xM} \varphi\xi$.
Hence $\xi$ is a $Y$-valued substitution such 
that $\models_{\xM} \varphi\xi$, and by assumption
we obtain $s\xi \Cab[\xE]{*} t\xi$.
Thus,
$(s\sigma)\theta
= s(\theta \circ \sigma) \Cab[\xE]{*} s\xi 
\Cab[\xE]{*} t\xi
\Cab[\xE]{*} t(\theta \circ \sigma)
= (t\sigma)\theta$.
We have shown 
for all $X$-valued substitutions $\theta$ with $\models_{\xM} (\varphi\sigma)\theta$
that $(s\sigma)\theta \Cab[\xE]{*} (t\sigma)\theta$ from which follows
$\fT \cec \CEqn{X}{s\sigma}{t\sigma}{\varphi\sigma}$.
\end{proof}
}
\ifthenelse{\boolean{OmitProofs}}{}{\ProofStability}

\begin{restatable}[general stability]{lemma}{LemmaStabilityGeneral}
\label{lem:stability general}
Let $\langle\xM, \xE \rangle$ be a CE-theory
and $\sigma$ a substitution such that
$\Dom(\sigma) \cap X = \varnothing$.
Then, if $\xE \cec \CEqn{X}{s}{t}{\varphi}$
then $\xE \cec \CEqn{X}{s\sigma}{t\sigma}{\varphi}$.
\end{restatable}

\begin{proof}[Proof (Sketch)]
Take any $X$-valued substitution $\delta$ such that $\models_{\xM} \varphi\delta$.
Take $\gamma = \delta \circ \sigma$.
From $\Dom(\sigma) \cap X = \varnothing$,
we have $\varphi\gamma = \varphi\delta$, and therefore,
$s\sigma\delta \Cab[\xE]{*} t\sigma\delta$ holds.
\end{proof}

\newcommand{\ProofStabilityGeneral}{
\ifthenelse{\boolean{OmitProofs}}{\LemmaStabilityGeneral*}{}
\begin{proof}
Assume $\xE \cec \CEqn{X}{s}{t}{\varphi}$
and a substitution $\sigma$ such that $\Dom(\sigma) \cap X = \varnothing$.
Let $\delta$ be an $X$-valued substitution such that $\models_{\xM} \varphi\delta$.
Consider the substitution $\gamma = \delta \circ \sigma$.
By $\Dom(\sigma) \cap X = \varnothing$,
we have 
$\gamma(x) = \delta(\sigma(x)) = \delta(x)$ for all $x \in X$.
Thus, $\gamma$ is $X$-valued and $\varphi\gamma = \varphi\delta$.
From our assumption $\xE \cec \CEqn{X}{s}{t}{\varphi}$,
we obtain $s\gamma \Cab[\xE]{*} t\gamma$ and further
$s\sigma\delta \Cab[\xE]{*} t\sigma\delta$.
We have that
$s\sigma\delta \Cab[\xE]{*} t\sigma\delta$ for 
any $X$-valued substitution with $\models_{\xM} \varphi\delta$,
which concludes 
$\xE \cec \CEqn{X}{s\sigma}{t\sigma}{\varphi}$.
\end{proof}
}
\ifthenelse{\boolean{OmitProofs}}{}{\ProofStabilityGeneral}

One may expect that $\xE \cec \CEqn{X}{s}{t}{\varphi}$ holds for
equivalent terms $s, t$ such that $\varphi \Rightarrow s = t$ is valid. In fact,
a more general result can be obtained.

\begin{restatable}[model consequence]{lemma}{LemmaModelConsequence}
\label{lem:model consequence}
    Let $\langle \xM, \xE \rangle$ be a CE-theory,
    $X \subseteq \xVTh$ a set of theory variables, 
    $s,t \in \xT(\xFTh,X)$,
    and $\varphi$ a logical constraint over $X$.
    If $\models_{\xM} ( \varphi\sigma \Rightarrow s\sigma = t \sigma )$
    holds for all $X$-valued substitutions $\sigma$,
    then $\xE \cec \CEqn{X}{s}{t}{\varphi}$.
\end{restatable}

\begin{proof}[Proof (Sketch)]
For any $X$-valued substitution $\sigma$ with $\models_{\xM} \varphi\sigma$,
we have $\inter{s\sigma}_{\xM} = \inter{t\sigma}_{\xM}$.
Then, use Lemma~\ref{lem:properties of calculation steps}
to obtain $s\sigma \Cab[\xE]{*} t\sigma$.
\end{proof}

\newcommand{\ProofModelConsequence}{
\ifthenelse{\boolean{OmitProofs}}{\LemmaModelConsequence*}{}
\begin{proof}
Consider a CE-theory $\fT = \langle \xM, \xE \rangle$,
a set $X \subseteq \xVTh$, terms $s,t \in \xT(\xFTh,X)$ and a logical
constraint $\varphi$ with $\Var(\varphi) \subseteq X$. 
Assume that
$\models_{\xM} (\varphi\sigma \Rightarrow s\sigma = t\sigma)$ for all $X$-valued
substitutions $\sigma$. 
Let $\sigma$ be an $X$-valued substitution with $\models_{\xM} \varphi\sigma$.
By our assumption, $\models_{\xM} s\sigma = t\sigma$ holds and thus also
$\inter{s\sigma}_{\xM} = \inter{t\sigma}_{\xM}$.
By Lemma~\ref{lem:properties of calculation steps}, 
we obtain $s\sigma \Cab[\xE]{*} t\sigma$.
We have that $s\sigma \Cab[\xE]{*} t\sigma$ 
for all $X$-valued substitution $\sigma$ with $\models_{\xM} \varphi\sigma$.
Therefore, $\xE \cec \CEqn{X}{s}{t}{\varphi}$.
\end{proof}
}
\ifthenelse{\boolean{OmitProofs}}{}{\ProofModelConsequence}

\begin{restatable}{corollary}{CorollaryGeneralTheoryRule}
\label{lem:general theory rule}
\label{cor:general theory rule}
Let $\langle\xM, \xE \rangle$ be a CE-theory,
$X\subseteq \xVTh$ a set of theory variables,
and $\varphi \mr{\Rightarrow} s = t$ a logical constraint over $X$
such that $\models_{\xM} ( \varphi \mr{\Rightarrow} s = t )$.
Then, $\xE \cec \CEqn{X}{s}{t}{\varphi}$.
\end{restatable}

\newcommand{\ProofGeneralTheoryRule}{
\ifthenelse{\boolean{OmitProofs}}{\CorollaryGeneralTheoryRule*}{}
\begin{proof}
    Let $X$ be a set of theory variables and
    $\varphi \mr{\Rightarrow} s = t$ a logical constraint over $X$.
    Suppose $\models_\xM (\varphi \mr{\Rightarrow} s = t)$.
    Let $\sigma$ be an $X$-valued substitution.
    Then, it follows from $\Var(\varphi,s,t) \subseteq X$ 
    that $\Var(\varphi \Rightarrow s = t) \cap \Dom(\sigma) \subseteq
    X \cap \Dom(\sigma) \subseteq \VDom(\sigma)$.
    Thus, by Lemma~\ref{lem:no-resp-subst-imply-false},
    $\models_{\xM} (\varphi\sigma \mr{\Rightarrow} s\sigma = t\sigma)$ holds.
    By Lemma~\ref{lem:model consequence} we
    obtain $\xE \cec \CEqn{X}{s}{t}{\varphi}$.
\end{proof}
}
\ifthenelse{\boolean{OmitProofs}}{}{\ProofGeneralTheoryRule}

\subsection{Relations to Conversion of Rewrite Steps}

In this subsection, we present characterizations of CE-validity from the
perspective of logically constrained rewriting with respect to equations.

\begin{restatable}{theorem}{TheoremRelationBetweenValidityAndConversionOfRewriting}
\label{thm:relation between validity and conversion of rewriting}
For a CE-theory $\langle\xM, \xE \rangle$,
$s \Cab[\xE]{*} t$
if and only if $\xE \cec \CEqn{\varnothing}{s}{t}{\mathsf{true}}$.
\end{restatable}

\begin{proof}[Proof (Sketch)]
We have $\xE \cec
\CEqn{\varnothing}{s}{t}{\mathsf{true}}$ if and only if $s\sigma \Cab[\xE]{*}
t\sigma$ for any $\varnothing$-valued substitution $\sigma$ such that $\sigma
\models \mathsf{true}$ if and only if $s\sigma \Cab[\xE]{*} t\sigma$ for any
substitution $\sigma$.
Thus, the claim follows by Lemma~\ref{lem:sym and closure properties of equational steps}.
\end{proof}

\newcommand{\ProofRelationBetweenValidityAndConversionOfRewriting}{
\ifthenelse{\boolean{OmitProofs}}{\TheoremRelationBetweenValidityAndConversionOfRewriting*}{}
\begin{proof}
Note that $\varnothing$-valued substitutions refer to all substitutions
and that 
$\sigma \models \mathsf{true}$ holds for any substitutions $\sigma$.
Thus, it follows from the definitions that: $\xE \cec
\CEqn{\varnothing}{s}{t}{\mathsf{true}}$ if and only if $s\sigma \Cab[\xE]{*}
t\sigma$ for any $\varnothing$-valued substitution $\sigma$ such that $\sigma
\models \mathsf{true}$ if and only if $s\sigma \Cab[\xE]{*} t\sigma$ for any
substitution $\sigma$.
Now, since the relation $\Cab[\xE]{*}$ is closed under substitutions by
Lemma~\ref{lem:sym and closure properties of equational steps}, $s\sigma
\Cab[\xE]{*} t\sigma$ for any substitution $\sigma$ if and only if $s
\Cab[\xE]{*} t$. Therefore, the claim follows.
\end{proof}
}
\ifthenelse{\boolean{OmitProofs}}{}{\ProofRelationBetweenValidityAndConversionOfRewriting}

We consider now the general case with a possibly non-empty set $X$ of
theory variables and a non-trivial constraint $\varphi \neq \mathsf{true}$ (also
$\neg \varphi$) for the CE $\CEqn{X}{s}{t}{\varphi}$. In this case, the
following partial characterization can be made by using the notion of trivial
CEs~\cite{SM23}. We can naturally extend the notion of trivial CEs
in~\cite{SM23} to our setting as follows: a CE $\CEqn{X}{s}{t}{\varphi}$ is said
to be \emph{trivial} if $s\sigma = t\sigma$ for any $X$-valued substitution
$\sigma$ such that $\models_{\xM} \varphi\sigma$.

\begin{restatable}{theorem}{TheoremProvingValidityByRewriting}
\label{thm:proving validity by rewriting}
Let $\langle\xM, \xE \rangle$ be a CE-theory, and $\CEqn{X}{s}{t}{\varphi}$ a CE. Suppose $s
\Cab[\xE]{*} s'$ and $t \Cab[\xE]{*} t'$ for some $s',t'$ such that
$\CEqn{X}{s'}{t'}{\varphi}$ is trivial. Then, $\xE \cec
\CEqn{X}{s}{t}{\varphi}$.
\end{restatable}

\begin{proof}[Proof (Sketch)]
Take an arbitrary $X$-valued substitution $\sigma$ such that $\models_{\xM}
\varphi\sigma$. Then, it follows from our assumptions that $s\sigma \Cab[\xE]{*}
s'\sigma = t'\sigma \Cab[\xE]{*} t\sigma$.
\end{proof}

\newcommand{\ProofProvingValidityByRewriting}{
\ifthenelse{\boolean{OmitProofs}}{\TheoremProvingValidityByRewriting*}{}
\begin{proof}
    Let $s',t'$ be terms such that 
    $s \Cab[\xE]{*} s'$, $t \Cab[\xE]{*} t'$, and 
    $\CEqn{X}{s'}{t'}{\varphi}$ is trivial.
    Let $\sigma$ be an arbitrary $X$-valued substitution
    such that $\models_{\xM} \varphi\sigma$.
    Then, from  $s \Cab[\xE]{*} s'$ and $t \Cab[\xE]{*} t'$,
    we have $s\sigma \Cab[\xE]{*} s'\sigma$ and $t\sigma \Cab[\xE]{*} t'\sigma$.
    Furthermore, from the triviality of $\CEqn{X}{s'}{t'}{\varphi}$,
    we know $s'\sigma = t'\sigma$.
    Thus, 
    $s\sigma \Cab[\xE]{*} s'\sigma = t'\sigma \Cab[\xE]{*} t\sigma$.
    This concludes $\xE \cec \CEqn{X}{s}{t}{\varphi}$.
\end{proof}
}
\ifthenelse{\boolean{OmitProofs}}{}{\ProofProvingValidityByRewriting}

Unfortunately, none of the CE-consequences in Example~\ref{exp:abs max} can be
handled by Theorems~\ref{thm:relation between validity and conversion of
rewriting} and~\ref{thm:proving validity by rewriting}.

\section{Proving CE-Validity}
\label{sec:proving-validity}

As mentioned in the previous section, CE-validity of a CE
$\CEqn{X}{s}{t}{\varphi}$ with respect to a CE-theory $\langle\xM, \xE \rangle$
is very tedious to be shown by the convertibility of $s$ and $t$ in
$\xE$. This motivates us to introduce another approach to reason about
CE-validity. In this section, we introduce an inference system for proving
CE-validity of CEs together with a discussion on its applicability and generality.

\subsection{Inference System \texorpdfstring{$\mathbf{CEC}_0$}{CEC\_0} and Its Soundness}

In this subsection, we introduce an inference system $\mathbf{CEC}_0$
(\emph{Constrained Equational Calculus for elementary steps}) for proving
CE-validity of CEs. We prove soundness of it, by which it is guaranteed
that all CEs $\CEqn{X}{s}{t}{\varphi}$ derivable from $\xE$ in the system
$\mathbf{CEC}_0$ are valid, i.e.\ $\xE \cec \CEqn{X}{s}{t}{\varphi}$.

\begin{definition}[derivation of $\mathbf{CEC}_0$]
\label{def:derivation of CEC0}
Let $\fT = \langle\xM, \xE \rangle$ be a CE-theory. The inference
system $\mathbf{CEC}_0$ consists of the inference rules given in
\Cref{fig:inf-rules}. We assume in the rules that $X,Y$ range over a (possibly
infinite) set of theory variables, $\varphi$ ranges over logical constraints.
Let $\CEqn{X}{s}{t}{\varphi}$ be a CE. We say that $\CEqn{X}{s}{t}{\varphi}$ is
derivable in $\mathbf{CEC}_0$ from $\xE$ (or $\CEqn{X}{s}{t}{\varphi}$ is a
consequence of $\xE$), written by $\xE \vdash_{\mathbf{CEC}_0}
\CEqn{X}{s}{t}{\varphi}$, if there exists a derivation of $\xE \vdash
\CEqn{X}{s}{t}{\varphi}$ in the system $\mathbf{CEC}_0$. When no confusion
arises, $\xE \vdash_{\mathbf{CEC}_0} \CEqn{X}{s}{t}{\varphi}$ is abbreviated as
$\xE \vdash \CEqn{X}{s}{t}{\varphi}$.
\end{definition}

\newcommand{\FigureInferenceRules}{%
\begin{figure}[tb]
\[
\begin{array}{c}
\begin{array}{lclc}
\deduce{\strut}{\deduce{\strut}{\mbox{\it  Refl}}}&
\infer[\Var(\varphi) \subseteq X]
  { \xE \vdash \CEqn{X}{s}{s}{\varphi}}
  { }
&
\deduce{\strut}{\deduce{\strut}{\mbox{\it Trans}}}&
\infer[]
  { \xE \vdash \CEqn{X}{s}{u}{\varphi}}
  { \xE \vdash \CEqn{X}{s}{t}{\varphi} 
    &  \xE \vdash \CEqn{X}{t}{u}{\varphi} }\\
\end{array}\\[5ex]
\begin{array}{lclc}
\deduce{\strut}{\deduce{\strut}{\mbox{\it  Sym}}}&
\infer[]
  { \xE \vdash \CEqn{X}{s}{t}{\varphi}}
  { \xE \vdash \CEqn{X}{t}{s}{\varphi}}
~~~~~\deduce{\strut}{\deduce{\strut}{\mbox{\it  Cong}}}&
\infer[]
  { \xE \vdash \CEqn{X}{f(s_1,\ldots,s_n)}{f(t_1,\ldots,t_n)}{\varphi}}
  { \xE \vdash \CEqn{X}{s_1}{t_1}{\varphi} 
    & \ldots & \xE \vdash \CEqn{X}{s_n}{t_n}{\varphi} }\\
\end{array}\\[5ex]
\begin{array}{l@{\qquad}c}
\deduce{\strut}{\mbox{\it  Rule}}&
\infer[(\CEqn{X}{\ell}{r}{\varphi}) \in \xE]
  { \xE \vdash \CEqn{X}{\ell}{r}{\varphi}}
  {}
\end{array}\\[3ex]
\begin{array}{l@{\qquad}c}
\deduce{\strut}{\mbox{\it Theory Instance}}&
\infer[\forall x \in Y.\ x\sigma \in \xT(\xFTh,X)]
{ \xE \vdash \CEqn{X}{s\sigma}{t\sigma}{\varphi\sigma}}
{ \xE \vdash \CEqn{Y}{s}{t}{\varphi}}\\
\end{array}
\\[3ex]
\begin{array}{l@{\qquad}c}
\deduce{\strut}{\mbox{\it  General Instance}}&
\infer[\Dom(\sigma) \cap X = \varnothing]
{ \xE \vdash \CEqn{X}{s\sigma}{t\sigma}{\varphi}}
{ \xE \vdash \CEqn{X}{s}{t}{\varphi}}\\
\end{array}
\\[3.5ex]
\begin{array}{lclc}
\deduce{\strut}{\deduce{\strut}{\mbox{\it  Weakening}}}& 
\infer[\begin{array}{l}
          \models_\xM (\varphi \mr{\Rightarrow} \psi),
           \Var(\varphi) \subseteq X
       \end{array}]
 { \xE \vdash \CEqn{X}{s}{t}{\varphi}}
 { \xE \vdash \CEqn{X}{s}{t}{\psi}} \\
\deduce{\strut}{\deduce{\strut}{\mbox{\it  Split}}}& 
\infer[]
  { \xE \vdash \CEqn{X}{s}{t}{\varphi \lor \psi}}
  { \xE \vdash \CEqn{X}{s}{t}{\varphi}
    &
    \xE \vdash \CEqn{X}{s}{t}{\psi}
    }\\
\end{array}\\[4ex]
\begin{array}{l@{\qquad}c}
\deduce{\strut}{\mbox{\it  Axiom}}&
\infer[\begin{array}{l}
\models_{\xM} (\varphi\sigma \Rightarrow s\sigma = t \sigma)
~\mbox{for all $\sigma$ s.t.\ $X \subseteq \VDom(\sigma)$},\\
\Var(\varphi) \subseteq X
\end{array}]
{ \xE \vdash \CEqn{X}{s}{t}{\varphi}}
{ }\\ 
\end{array}\\[4ex]
\begin{array}{l@{\qquad}c}
\deduce{\strut}{\mbox{\it  Abst}}&
\infer[\begin{array}{l}
       \models_{\xM} (\varphi  \Rightarrow \bigwedge_{x \in X} x = \sigma(x)),\\
       \Var(\varphi) \subseteq X, 
       \left(\bigcup_{x \in X}\Var(\sigma(x))\right) \subseteq X
       \end{array}]
  { \xE \vdash \CEqn{X}{s}{t}{\varphi}}
  { \xE \vdash \CEqn{X}{s\sigma}{t\sigma}{\varphi\sigma}}\\
\end{array}\\[3.5ex]
\begin{array}{l@{\qquad}c}
\deduce{\strut}{\mbox{\it  Enlarge}}&
\infer[\Var(s,t) \cap (Y \setminus X) = \varnothing, ~\Var(\varphi) \subseteq X]
 { \xE \vdash \CEqn{X}{s}{t}{\varphi}}
 { \xE \vdash \CEqn{Y}{s}{t}{\varphi}}\\
\end{array}
\end{array}
\]
\caption{Inference rules of $\mathbf{CEC}_0$.}
\label{fig:inf-rules}
\end{figure}
}

\FigureInferenceRules

We proceed with intuitive explanations of each rule.
The rules
\textit{Refl}, \textit{Trans}, \textit{Sym}, \textit{Cong}, and \textit{Rule}
are counterparts of the inference rules used in equational logic.

In order to handle instantiations, we consider two cases, namely
\textit{Theory Instance} and \textit{General Instance}. The former rule
covers instantiations affecting the logical constraint whereas the
latter covers the case not affecting it. 

The \textit{Weakening} and
\textit{Split} rules handle logical reasoning in constraints. The
\textit{Weakening} rule logically weakens the constraint equation by
strengthening its constraint. Note here the direction of the entailment
$\varphi \Rightarrow \psi$ in the side condition: the rule is sound because
the constrained equation is valid under the stronger constraint
$\varphi$ if the equation is valid under the weaker constraint $\psi$.
Since some rules, like \textit{Cong} and \textit{Trans}, have
premises with equality constraints, it may be required to first apply the
\textit{Weakening} rule to synchronize the constraints before using these rules.
On the other hand, in the \textit{Split} rule, the constraint of the
conclusion $\varphi \lor \psi$ is logically weaker than the independent
ones, $\varphi$ and $\psi$, in each premise. The inference is still
sound as it only joins two premises. Using the \textit{Split} rule, one can
perform reasoning based on case analysis. 

The \textit{Axiom} rule
makes it possible to use equational consequences entailed in the constraint
part of equational reasoning. 
The \textit{Abst} rule incorporates
consequences entailed in the constraint part in a different way, that is, to
infer a possible abstraction of the equation.

The rules \textit{Enlarge}
and \textit{Restrict} are used to adjust the set of instantiated variables (the
``$\Pi X$'' part of CEs), with the proviso that it does not affect the validity.
Note that, in \textit{Enlarge}, $Y \subseteq X$ implies the side condition
$\Var(s,t) \cap (Y \setminus X) = \varnothing$. We also want to remark that
despite its name, the restriction $X \subseteq Y$ can be added to
\textit{Enlarge}, provided that removed variables are not used in $s,t$ (the
side condition $\Var(s,t) \cap (Y \setminus X) = \varnothing$ has to be
satisfied).

\begin{restatable}{lemma}{LemmaPreservationOfCE}
\label{lem:preservation of CE}
Let $\langle \xM, \xE \rangle$ be a CE-theory.
If $\xE \vdash_{\mathbf{CEC}_0} \CEqn{X}{s}{t}{\varphi}$,
then $\CEqn{X}{s}{t}{\varphi}$ is a CE.
\end{restatable}

\begin{proof}[Proof (Sketch)]
The proof proceeds by induction on the derivation
of $\xE \vdash \CEqn{X}{s}{t}{\varphi}$ that 
(1) $s,t \in \xT(\xF,\xV)$ have the same sort,
(2) $\varphi$ is a logical constraint,
(3) $X \subseteq \xVTh$, and
(4) $\Var(\varphi) \subseteq X$.
\end{proof}

\newcommand{\ProofPreservationOfCE}{
\ifthenelse{\boolean{OmitProofs}}{\LemmaPreservationOfCE*}{}
\begin{proof}
We prove this claim by induction on the derivation.
We make a case analysis depending on the last rule that have been applied.
We have to check that for the conclusion of the inference rule,
say $\xE \vdash \CEqn{X}{s}{t}{\varphi}$, the following 
conditions are met:
(1) $s,t \in \xT(\xF,\xV)$ have the same sort,
(2) $\varphi$ is a logical constraint,
(3) $X \subseteq \xVTh$, and
(4) $\Var(\varphi) \subseteq X$.
\begin{itemize}
\item Case where \textit{Refl} is applied.
    (1) is trivial.
    (2), (3) are guaranteed by our convention on the inference
    rules (Definition~\ref{def:derivation of CEC0}).
    (4) is guaranteed by the side condition of the rule.

\item Case where \textit{Trans} is applied.
    (1) follows from the induction hypotheses.
    (2)--(4) follow from (one of) the induction hypotheses.

\item Case where \textit{Sym} is applied.
    (1)--(4) follow from the induction hypothesis.

\item Case where \textit{Cong} is applied.
    (1) immediately follows from the induction hypotheses.
    (2)--(4) follow from (one of) the induction hypotheses.

\item Case where \textit{Rule} is applied.
   (1)--(4) hold because $\xE$ is an LCES (i.e.\ a set of CEs.)

\item Case where \textit{Instance} is applied.
    (1) immediately follows from the induction hypothesis.
    (3) By our convention on the inference
    rules (Definition~\ref{def:derivation of CEC0}),
    $X$ is a set of theory variables.
    (4) We have $\Var(\varphi) \subseteq Y$ by induction hypothesis,
    and by the side condition of the inference rule,
    we have $x\sigma \in \xT(\xFTh,X)$ for all $x \in Y$.
    These together imply that 
    $\Var(\varphi\sigma) \subseteq X$.
    (2) By the previous (3) and (4),
    $\varphi\sigma \in \xT(\xFTh,\xVTh)$.
    It remains to show that the sort of $\varphi\sigma$
    is $\m{Bool}$, but hold as so is $\varphi$, 
    and the substitution does not change the sort of terms.

\item Case where \textit{Term Instance} is applied.
    (1)--(4) follow from the induction hypothesis.

\item Case where \textit{Weakening} is applied.
    (1) and (3) follow from the induction hypothesis.
    (2) follows from the side condition of the inference rule,
    as $\models_{\xM} (\varphi \Rightarrow \psi)$ guarantee
    that $\varphi \Rightarrow \psi$ is a logical constraint,
    hence so is the $\varphi$.
    (4) is guaranteed by the side condition of the inference rule.

\item Case where \textit{Split} is applied.
    (1) and (3) follow from (one of) the induction hypotheses.
    (2) and (4) immediately follow from the induction hypotheses.

\item Case where \textit{Axiom} is applied.
    (1) We have that $s\sigma  = t\sigma$ is a logical constraint,
        which is guaranteed by the side condition.
        Then, since $=$ is an abbreviation for 
        $=^\tau\colon  \tau \times \tau \to \m{Bool}$ for some $\tau$,
        the claim follows.
    (3) By our convention on the inference
    rules (Definition~\ref{def:derivation of CEC0}),
    $X$ is a set of theory variables.
    (4) is guaranteed by the side condition of the inference rule.
    (2) We have that $\varphi\sigma$ is a logical constraint,
        which is guaranteed by the side condition.
        Thus, the only possibility of $\varphi$ not being a logical constraint,
        is to have a term variable in it. But
        this possibility is excluded by (4).

    \item Case where \textit{Abst} is applied.
    (1) immediately follows from the induction hypothesis.
    (2) follows from the side condition of the inference rule.
    (3) By our convention on the inference
    rules (Definition~\ref{def:derivation of CEC0}),
    $X$ is a set of theory variables.
    (4) is guaranteed by the side condition of the inference rule.


    \item Case where \textit{Enlarge} is applied.
    (1) and (2) follow from the induction hypothesis.
    (3) By our convention on the inference
    rules (Definition~\ref{def:derivation of CEC0}),
    $X$ is a set of theory variables.
    (4) is guaranteed by the side condition of the inference rule.
\end{itemize}
\end{proof}
}
\ifthenelse{\boolean{OmitProofs}}{}{\ProofPreservationOfCE}

Below we present some examples of derivations which cover all of our
inference rules at least once. In the following example we denote 
\textit{Theory Instance} by \textit{TInst} and \textit{General Instance}
by \textit{GInst} accordingly.

\begin{example}
Let $\langle\xM, \xE \rangle$ be the CE-theory
given in Example~\ref{exa:group with exponentials}.
Below, we present a derivation of 
$\xE \vdash \CEqn{\SET{ n }}{\m{exp}(x,n) * \m{exp}(x,-n)}{\m{e}}{}$.
{\small
\[
\infer[\!\!\it Trans]
 {\xE \vdash \CEqn{\SET{ n }}
                  {\m{exp}(x,n) * \m{exp}(x,-n)}
                  {\m{e}}
                  {}}
 {\infer[\!\it TInst]
   {\xE \vdash \CEqn{\SET{ n }}
                   {\m{exp}(x,n) * \m{exp}(x,-n)}
                   {\m{exp}(x,n+(-n))}
                   {}}
   {\infer[\!\it Rule]
     {\xE \vdash \CEqn{\SET{ n,m }}
                   {\m{exp}(x,n) * \m{exp}(x,m)}
                   {\m{exp}(x,n+m)}
                   {}}
     {}
   }
  &
  \deduce[]
   {\xE \vdash \CEqn{\SET{ n }}
                   {\m{exp}(x,n+(-n))}
                   {\m{e}}
                   {}}
   {(1)}}
\]}
where (1) is
{\small
\[
\infer[\it Trans]
 {\xE \vdash \CEqn{\SET{ n }}
                  {\m{exp}(x,n+(-n))}
                  {\m{e}}
                  {}}
 {\infer[\it Cong]
   {\xE \vdash \CEqn{\SET{ n }}
                   {\m{exp}(x,n+(-n))}
                   {\m{exp}(x,\m{0})}
                   {}}
   {\infer[\it Enlarge]
     {\xE \vdash \CEqn{\SET{ n }}{x}{x}{}}
     {\infer[\it Refl]
        {\xE \vdash \CEqn{}{x}{x}{}}
        {}}
    &
    \infer[\it Axiom]
     {\xE \vdash \CEqn{\SET{ n }}{n+(-n)}{\m{0}}{}}
     {}}
  &
  \infer[\it Enlarge]
     {\xE \vdash \CEqn{\SET{ n }}
                      {\m{exp}(x,\m{0})}
                      {\m{e}}
                      {}}
     {\infer[\it Rule]
        {\xE \vdash \CEqn{}
                      {\m{exp}(x,\m{0})}
                      {\m{e}}
                      {}}
         {}}}
\]}
Here, $\xE \vdash \CEqn{\SET{ n }}{n+(-n)}{\m{0}}{}$ is derived by the
\textit{Axiom} rule, because of $\models_{\xM} \sigma(n+(-n)) = \m{0}$ holds for
all $\SET{n}$-valued substitutions $\sigma$.
\end{example}

\begin{example}
Let $\langle\xM, \xE \rangle$ be the CE-theory given in Example~\ref{exp:list
functions}. Below, we present a derivation of 
$\xE \vdash \CEqn{\SET{ n }}{\m{nth}(x{::}y{::}zs,n+2)}{\m{nth}(zs,n)}{n>0}$.
{\small
\[
\infer[\it Trans]
 {\xE \vdash \CEqn{\SET{ n }}{\m{nth}(x{::}y{::}zs,n+2)}{\m{nth}(zs,n)}{n>0}}
 {(2)
  &
  \infer[\it GInst]
   {\xE \vdash \CEqn{\SET{ n }}{\m{nth}(y{::}zs,n+1)}{\m{nth}(zs,n)}{n>0}}
   {\infer[\it Weaken]
     {\xE \vdash \CEqn{\SET{ n }}{\m{nth}(x{::}xs,n+1)}{\m{nth}(xs,n)}{n>0}}
     {\infer[\it Trans]
       {\xE \vdash \CEqn{\SET{ n }}{\m{nth}(x{::}xs,n+1)}{\m{nth}(xs,n)}{n+1>0}}
       {\infer[\it TInst]
         {\xE \vdash \CEqn{\SET{ n }}{\m{nth}(x{::}xs,n+1)}{\m{nth}(xs,(n+1)-1)}{n+1>0}}
         {\infer[\it Rule]
           {\xE \vdash \CEqn{\SET{ n }}{\m{nth}(x{::}xs,n)}{\m{nth}(xs,n-1)}{n>0}}
            {}
         }
        &
         (1)
       }
     }
   }
 }
\]}
where (1) is
{\small
\[
\infer[\it Weaken]
  {\xE \vdash \CEqn{\SET{ n }}{\m{nth}(xs,((n+1)-1)}{\m{nth}(xs,n)}{n+1>0}}
     {\infer[\it Cong]
        {\xE \vdash \CEqn{\SET{ n }}{\m{nth}(xs,((n+1)-1)}{\m{nth}(xs,n)}{}}
        {\infer[\it Refl]
          {\xE \vdash \CEqn{\SET{ n }}{xs}{xs}{}}
          {}
         &
         \infer[\it Axiom]
             {\xE \vdash \CEqn{\SET{ n }}{(n+1)-1}{n}{}}
             {}
         }
       }
\]}
and (2) is 
{\small
\[
 \infer[\it GInst]
   {\xE \vdash \CEqn{\SET{ n }}{\m{nth}(x::y::zs,n+2)}{\m{nth}(y::zs,n+1)}{n>0}}
   {\infer[\it Weaken]
     {\xE \vdash \CEqn{\SET{ n }}{\m{nth}(x::ys,n+2)}{\m{nth}(ys,n+1)}{n>0}}
     {\infer[\it Trans]
       {\xE \vdash \CEqn{\SET{ n }}{\m{nth}(x::xs,n+2)}{\m{nth}(xs,n+1)}{n+2>0}}
       {\infer[\it TInst]
           {\xE \vdash \CEqn{\SET{ n }}{\m{nth}(x::xs,n+2)}{\m{nth}(xs,(n+2)-1)}{n+2>0}}
           {\infer[\it Rule]
            {\xE \vdash \CEqn{\SET{ n }}{\m{nth}(x::xs,n)}{\m{nth}(xs,n-1)}{n>0}}
            {}}
        &
        (3)}
     } 
   }
\]}
where (3) is
{\small
\[
\infer[\it Weaken]
 {\xE \vdash \CEqn{\SET{ n }}{\m{nth}(xs,(n+2)-1)}{\m{nth}(xs,n+1)}{n+2>0}}
     {\infer[\it Cong]
        {\xE \vdash \CEqn{\SET{ n }}{\m{nth}(xs,((n+2)-1)}{\m{nth}(xs,n+1)}{}}
        {\infer[\it Refl]
          {\xE \vdash \CEqn{\SET{ n }}{xs}{xs}{}}
          {}
         &
         \infer[\it Axiom]
             {\xE \vdash \CEqn{\SET{ n }}{(n+2)-1}{n+1}{}}
             {}
         }
       }
\]}
\end{example}

In the next example, we illustrate usages of the \textit{Split} rule and the
\textit{Abst} rule.

\begin{example}
    Let $\mathsf{f}\colon \Bool \to \Int \in \xFTe$.
    Let $\xE = \SET{ 
        \CEqn{\varnothing}{\mathsf{f}(\mathsf{true})}{\mathsf{0}}{\mathsf{true}}, \,
        \CEqn{\varnothing}{\mathsf{f}(\mathsf{false})}{\mathsf{0}}{\mathsf{true}}
    }$.
    Then we have $\mathsf{f}(\mathsf{true}) \leftrightarrow_\xE \mathsf{0}$
    and $\mathsf{f}(\mathsf{false}) \leftrightarrow_\xE \mathsf{0}$.
    Thus, for all $\SET{ x }$-valued substitutions $\sigma$,
    we have $\mathsf{f}(x)\sigma \Cab[\xE]{*} \mathsf{0}\sigma$.
{\small
\[
\infer[\it Weaken]
 {\xE \vdash \CEqn{\SET{ x }}{\mathsf{f}(x)}{\mathsf{0}}{\mathsf{true}}}
 {\infer[Split]
     {\xE \vdash \CEqn{\SET{ x }}{\mathsf{f}(x)}{\mathsf{0}}{x = \mathsf{true} \lor x = \mathsf{false}}}
     {\infer[\it Abst]
       {\xE \vdash \CEqn{\SET{ x }}{\mathsf{f}(x)}{\mathsf{0}}{x = \mathsf{true}}}
       {\infer[\it Weaken]
         {\xE \vdash \CEqn{\varnothing}{\mathsf{f}(\mathsf{true})}{\mathsf{0}}{\mathsf{true} =     \mathsf{true}}}
         {\infer[\it Rule]
             {\xE \vdash \CEqn{\varnothing}{\mathsf{f}(\mathsf{true})}{\mathsf{0}}{}}
             {}}}
      &
      \infer[\it Abst]
       {\xE \vdash \CEqn{\SET{ x }}{\mathsf{f}(x)}{\mathsf{0}}{x = \mathsf{false}}}
       {\infer[\it Weaken]
         {\xE \vdash \CEqn{\varnothing}{\mathsf{f}(\mathsf{false})}{\mathsf{0}}{\mathsf{false} =     \mathsf{false}}}
         {\infer[\it Rule]
            {\xE \vdash \CEqn{\varnothing}{\mathsf{f}(\mathsf{false})}{\mathsf{0}}{}}
            {}}}}}
\]}
\end{example}

We now present soundness of the system $\mathbf{CEC}_0$ with respect to
CE-validity.

\begin{restatable}[soundness of the system $\mathbf{CEC}_0$]{theorem}{TheoremSoundnessOfTheSystemCEC}
\label{prop:rewriting interpretation}
\label{thm:soundness of the system CEC0}
Let $\langle\xM, \xE \rangle$ be a CE-theory and
$\CEqn{X}{s}{t}{\varphi}$ a CE. If $\xE \vdash_{\mathbf{CEC}_0}
\CEqn{X}{s}{t}{\varphi}$, then $\xE \cec \CEqn{X}{s}{t}{\varphi}$.
\end{restatable}

\begin{proof}[Proof (Sketch)]
The proof proceeds by induction on the derivation.
The cases for all rules except for \textit{Abst} and \textit{Enlarge}
follow by Lemmas~\ref{lem:rule}--\ref{lem:model consequence} and well-known
properties in logic.
For the case \textit{Abst}, suppose that $\xE \vdash \CEqn{X}{s}{t}{\varphi}$ is
derived from $\xE \vdash \CEqn{X}{s\sigma}{t\sigma}{\varphi\sigma}$ as given in
Figure~\ref{fig:inf-rules}. Let $\rho$ be an $X$-valued substitution such that
$\models_{\xM} \varphi\rho$. Then by the side conditions we have $\inter{\rho(x)} =
\inter{\rho(\sigma(x))}$ for any $x \in X$. Hence, $\models_{\xM}
\varphi\sigma\rho$, and $s\sigma\rho \Cab[\xE]{*} t\sigma\rho$ by the induction
hypothesis. Then by $\Var(s,t) \subseteq X$, we have $s\rho \Cab[\xE]{*}
s\sigma\rho \Cab[\xE]{*} t\sigma\rho \Cab[\xE]{*} t\rho$.
For the case \textit{Enlarge}, suppose that $\xE \vdash \CEqn{X}{s}{t}{\varphi}$
is derived from $\xE \vdash \CEqn{Y}{s}{t}{\varphi}$ as given in
Figure~\ref{fig:inf-rules}. Let $\delta$ be an $X$-valued substitution such that
$\models_{\xM} \varphi\delta$. Define a substitution $\delta'$ as follows:
$\delta'(z^\tau) = c^\tau$ if $x \in Y \setminus X$, and $\delta'(x) =
\delta(x)$ otherwise, where $c^\tau$ is (arbitrarily) taken from $\Val^{\tau}$.
Clearly, $\delta'$ is $Y$-valued. We also have $s\delta' = s\delta$, $t\delta' =
t\delta$, and $\varphi\delta' = \varphi\delta$ by the side conditions. Thus,
$s\delta = s\delta' \Cab[\xE]{*} t\delta' = t\delta$ using the induction
hypothesis.
\end{proof}

\newcommand{\ProofSoundnessOfTheSystemCEC}{
\ifthenelse{\boolean{OmitProofs}}{\TheoremSoundnessOfTheSystemCEC*}{}
\begin{proof}
We prove this claim by induction on the derivation. We make a case analysis
depending on which inference rule is applied at the last derivation step.
\begin{itemize}
\item Case where \textit{Refl} is applied.
    Suppose that $\xE \vdash \CEqn{X}{s}{s}{\varphi}$ is derived as given in 
    Figure~\ref{fig:inf-rules}.
    By Lemma~\ref{lem:congruence}, we have $\xE \cec \CEqn{X}{s}{s}{\varphi}$.
    This shows the claim.

    \item Case where \textit{Trans} is applied.
    Suppose that $\xE \vdash \CEqn{X}{s}{u}{\varphi}$ is derived
    from $\xE \vdash \CEqn{X}{s}{t}{\varphi}$ and $\xE \vdash \CEqn{X}{t}{u}{\varphi}$ 
    as given in Figure~\ref{fig:inf-rules}.
    By the induction hypotheses, we have $\xE \cec \CEqn{X}{s}{t}{\varphi}$ and $\xE \cec \CEqn{X}{t}{u}{\varphi}$.
    Then, it follows from Lemma~\ref{lem:congruence} that $\xE \cec \CEqn{X}{s}{u}{\varphi}$.

    \item Case where \textit{Sym} is applied.
    Suppose that $\xE \vdash \CEqn{X}{s}{t}{\varphi}$ is derived
    from $\xE \vdash \CEqn{X}{t}{s}{\varphi}$
    as given in Figure~\ref{fig:inf-rules}.
    By the induction hypotheses, we have $\xE \cec \CEqn{X}{t}{s}{\varphi}$.
    Then, it follows from Lemma~\ref{lem:congruence} that $\xE \cec \CEqn{X}{s}{t}{\varphi}$.

    \item Case where \textit{Cong} is applied.
    Suppose that 
    $\xE \vdash \CEqn{X}{f(s_1,\ldots,s_n)}{f(t_1,\ldots,t_n)}{\varphi}$ 
    is derived from $\xE \vdash \CEqn{X}{s_1}{t_1}{\varphi},
    \ldots,\xE \vdash \CEqn{X}{s_n}{t_n}{\varphi}$
    as given in Figure~\ref{fig:inf-rules}.
    By the induction hypotheses, 
    $\xE \cec \CEqn{X}{s_i}{t_i}{\varphi}$ for all $1 \leqslant i \leqslant n$.
    Then, it follows from Lemma~\ref{lem:congruence} that 
    $\xE \cec \CEqn{X}{f(s_1,\ldots,s_n)}{f(t_1,\ldots,t_n)}{\varphi}$.

    \item Case where \textit{Rule} is applied.
    Suppose that $\xE \vdash \CEqn{X}{l}{r}{\varphi}$ is derived 
    as given in Figure~\ref{fig:inf-rules}.
    By the side condition, $(\CEqn{X}{l}{r}{\varphi}) \in \xE$.
    Thus, one obtains from Lemma~\ref{lem:rule}
    that $\xE \cec \CEqn{X}{l}{r}{\varphi}$.

    \item Case where \textit{Theory Instance} is applied.
    Suppose that $\xE \vdash \CEqn{X}{s\sigma}{t\sigma}{\varphi\sigma}$ is
    derived from $\xE \vdash \CEqn{X}{s}{t}{\varphi}$,
    as given in Figure~\ref{fig:inf-rules}.
    By the induction hypothesis, we have $\xE \cec \CEqn{X}{s}{t}{\varphi}$.
    Then, because of the side condition, one can apply Lemma~\ref{lem:stability}, 
    so that we have $\xE \cec \CEqn{X}{s\sigma}{t\sigma}{\varphi\sigma}$.

    \item Case where \textit{General Instance} is applied.
    Suppose that $\xE \vdash \CEqn{X}{s\sigma}{t\sigma}{\varphi}$ is
    derived from $\xE \vdash \CEqn{X}{s}{t}{\varphi}$,
    as given in Figure~\ref{fig:inf-rules}.
    By the induction hypothesis, we have $\xE \cec \CEqn{X}{s}{t}{\varphi}$.
    Then, because of the side condition, 
    one can apply Lemma~\ref{lem:stability general}, 
    so that we have $\xE \cec \CEqn{X}{s\sigma}{t\sigma}{\varphi}$.

    \item 
    Case where \textit{Weakening} is applied.
    Suppose that $\xE \vdash \CEqn{X}{s}{t}{\psi}$ is derived
    from $\xE \vdash \CEqn{X}{s}{t}{\varphi}$
    as given in Figure~\ref{fig:inf-rules}.
    Let $\sigma$ be an $X$-valued substitution such that $\models_{\xM} \varphi\sigma$.
    As we have $\models_{\xM} (\varphi \Rightarrow \psi)$ by the side condition,
    we have $\models_{\xM} (\varphi\sigma \Rightarrow \psi\sigma)$
    (also because of $\Var(\varphi) \subseteq X$ and $\sigma$ is $X$-valued),
    and hence $\models_{\xM} \psi\sigma$.
    By the induction hypothesis, we have $\xE \cec \CEqn{X}{s}{t}{\psi}$,
    that is, $s\sigma' \Cab[\xE]{*} t\sigma'$ holds for any 
    $X$-valued substitution such that $\models_{\xM} \psi\sigma'$.
    Thus, one obtains $s\sigma \Cab[\xE]{*} t\sigma$.
    We now have shown that  $s\sigma \Cab[\xE]{*} t\sigma$ holds
    for any $X$-valued substitution $\sigma$ 
    such that $\models_{\xM} \varphi\sigma$,
    and this concludes $\xE \cec \CEqn{X}{s}{t}{\varphi}$.

    \item 
    Case where \textit{Split} is applied.
    Suppose that $\xE \vdash \CEqn{X}{s}{t}{\varphi \lor \psi}$ is derived
    from $\xE \vdash \CEqn{X}{s}{t}{\varphi}$
    and $\xE \vdash \CEqn{X}{s}{t}{\psi}$
    as given in Figure~\ref{fig:inf-rules}.
    Let $\sigma$ be an $X$-valued substitution such that 
    $\models_{\xM} (\varphi \lor \psi)\sigma$,
    i.e.\ $\models_{\xM} (\varphi\sigma \lor \psi\sigma)$.
    Then, $\models_{\xM} \varphi\sigma$ or $\models_{\xM} \psi\sigma$ holds
    (note here $\varphi\sigma, \psi\sigma$ are ground theory terms).
    By the induction hypotheses, we have
    $\xE \cec \CEqn{X}{s}{t}{\varphi}$ and
    $\xE \cec \CEqn{X}{s}{t}{\psi}$,
    and thus, in either case, $s\sigma \Cab[\xE]{*} t\sigma$ holds.
    Therefore, we have that $\xE \cec \CEqn{X}{s}{t}{\varphi \lor \psi}$.
    
    \item 
    Case where \textit{Axiom} is applied.
    Suppose that $\xE \vdash \CEqn{X}{s}{t}{\varphi}$ is derived 
    as given in Figure~\ref{fig:inf-rules}.
    By the side conditions, one can apply Lemma~\ref{lem:model consequence},
    so that $\xE \cec \CEqn{X}{s}{t}{\varphi}$ holds.

    \item 
    Case where \textit{Abst} is applied.
    Suppose that $\xE \vdash \CEqn{X}{s}{t}{\varphi}$ is derived
    from $\xE \vdash \CEqn{X}{s\sigma}{t\sigma}{\varphi\sigma}$
    as given in Figure~\ref{fig:inf-rules}.
    Let $\rho$ be an $X$-valued substitution such that $\models_{\xM} \varphi\rho$.
    As we have $\models_{\xM} (\varphi \Rightarrow \bigvee_{x\in X} x = \sigma(x))$,
    we have $\models_{\xM} (\varphi\rho \Rightarrow 
    \bigvee_{x\in X} \rho(x)= \rho(\sigma(x)))$,
    and thus $\models_{\xM} \bigvee_{x\in X} \rho(x)= \rho(\sigma(x)))$,
    i.e.\ $\inter{\rho(x)} = \inter{\rho(\sigma(x))}$ for any $x \in X$
    (note here that $\rho(\sigma(x))$ is a ground theory term
    by the side condition $\bigcup_{x \in X} \Var(\sigma(x)) \subseteq X$).
    Thus, it follows from $\models_{\xM} \varphi\rho$
    that $\models_{\xM} \varphi\sigma\rho$, as $\Var(\varphi) \subseteq X$.
    By the induction hypothesis, we have
    $\xE \cec \CEqn{X}{s\sigma}{t\sigma}{\varphi\sigma}$,
    and thus, $s\sigma\rho \Cab[\xE]{*} t\sigma\rho$ holds.
    Hence, from $\inter{\rho(x)} = \inter{\rho(\sigma(x))}$ for any $x \in X$
    (and $\Var(s,t) \subseteq X$),
    one obtains $s\rho \Cab[\xE]{*} t\rho$.
    Therefore, we now have shown that  $s\rho \Cab[\xE]{*} t\rho$ holds
    for any $X$-valued substitution $\rho$ 
    such that $\models_{\xM} \varphi\rho$,
    and this concludes $\xE \cec \CEqn{X}{s}{t}{\varphi}$.

    \item Case where \textit{Enlarge} is applied.
    Suppose that $\xE \vdash \CEqn{X}{s}{t}{\varphi}$ is derived
    from $\xE \vdash \CEqn{Y}{s}{t}{\varphi}$
    as given in Figure~\ref{fig:inf-rules}.
    Let $\delta$ be an $X$-valued substitution such that 
    $\models_{\xM} \varphi\delta$.
    Define a substitution $\delta'$ as follows:
    $\delta'(z^\tau) = c^\tau$ if $x \in Y \setminus X$, and
    $\delta'(x) = \delta(x)$ otherwise,
    where $c^\tau$ is (arbitrarily) taken from $\Val^{\tau}$.
    Apparently, $\delta'$ is $Y$-valued by the definition,
    as $\delta$ is $X$-valued.
    Furthermore, since $\Var(\varphi) \subseteq X$ by the side condition,
    we have $\varphi\delta' = \varphi\delta$, and thus,
    $\models_{\xM} \varphi\delta'$ also holds.
    By the induction hypothesis, 
    we have $\xE \cec \CEqn{Y}{s}{t}{\varphi}$,
    and thus, $s\delta' \Cab[\xE]{*} t\delta'$.
    From the side condition that 
    $\Var(s,t) \cap (Y \setminus X) = \varnothing$,
    we have $s\delta' = s\delta$ and $t\delta' = t\delta$.
    Thus, $s\delta \Cab[\xE]{*} t\delta$.
    Therefore, we now have shown that  $s\delta \Cab[\xE]{*} t\delta$ holds
    for any $X$-valued substitution $\delta$ 
    such that $\models_{\xM} \varphi\delta$,
    and this concludes $\xE \cec \CEqn{X}{s}{t}{\varphi}$.

\end{itemize}
\end{proof}
}
\ifthenelse{\boolean{OmitProofs}}{}{\ProofSoundnessOfTheSystemCEC}

\begin{remark}
We remark that some of our inference rules utilize validity in the model.
Thus, $\mathbf{CEC}_0$ does not have convenient properties like recursive
enumerability of its theorems like we are used to from other formal systems.
\end{remark}

\subsection{Partial Completeness of \texorpdfstring{$\mathbf{CEC}_0$}{CEC\_0}}

In this subsection, we present some results regarding the completeness property
of $\mathbf{CEC}_0$.
  
\begin{restatable}{lemma}{LemmaAbstractingReflexivity}
\label{lem: abstracting reflexivity}
Let $\langle\xM, \xE \rangle$ be a CE-theory and
$\CEqn{X}{s}{t}{\varphi}$ a CE such that $\varphi$ is satisfiable.
Suppose $s\sigma = t\sigma$ for all $X$-valued substitutions $\sigma$ such that
$\models_{\xM} \varphi\sigma$. Then $\xE \vdash \CEqn{X}{s}{t}{\varphi}$.
\end{restatable}

\begin{proof}[Proof (Sketch)]
The case $s,t \in \xT(\xFTh,X)$ follows by the assumption using the
\textit{Axiom} rule. Next, we consider the case $s = x \in \xV$ with $x \notin
X$. In this case, we can derive $t = x$ using the assumption, and the case
follows using the \textit{Refl} rule. For the general case, from the assumption,
one can let $s = C[s_1,\ldots,s_n]$ and $t = C[t_1,\ldots,t_n]$ for a multi-hole
context $C$ and terms $s_i,t_i$ ($1 \le i \le n)$ such that either one of $s_i$
or $t_i$ is a variable. Thus, by the previous cases, we know that $\xE \vdash
\CEqn{X}{s_i}{t_i}{\varphi}$ for each $1 \leqslant i \leqslant n$, possibly
using the \textit{Sym} rule. Thus, the claim is obtained using \textit{Refl},
\textit{Trans} and  \textit{Cong} rules.
\end{proof}

\newcommand{\ProofAbstractingReflexivity}{
\ifthenelse{\boolean{OmitProofs}}{\LemmaAbstractingReflexivity*}{}
\begin{proof}
First, we consider the case $s,t \in \xT(\xFTh,X)$. In this case, the claim is
obtained easily from the \textit{Axiom} rule as follows. Suppose $\sigma$ is an
$X$-valued substitution. If $\models_{\xM} \varphi\sigma$, then it follows from
the assumption that $s\sigma = t\sigma$. Then, by $s,t \in \xT(\xFTh,X)$ and $X
\subseteq \VDom(\sigma)$, we have $\inter{s\sigma} = \inter{t\sigma}$, and
hence, $\models_{\xM} (\varphi\sigma \Rightarrow s\sigma = t\sigma)$ holds.
Suppose $\not\models_{\xM} \varphi\sigma$. In this case, $\models_{\xM}
(\varphi\sigma \Rightarrow s\sigma = t\sigma)$ holds trivially. All in all,
$\models_{\xM} (\varphi\sigma \Rightarrow s\sigma = t\sigma)$ for any $X$-valued
substitution $\sigma$. Therefore, by \textit{Axiom} rule, $\xE \vdash
\CEqn{X}{s}{t}{\varphi}$.

Next, we consider the case $s = x \in \xV$. The case $x \in X$ is included in
the case above. So, assume $x \notin X$. By our assumption, we exclude that case
that $\varphi$ is not satisfiable. Let $\sigma$ be a $X$-valued substitution
such that $\models_{\xM} \varphi\sigma$. Furthermore, take $\sigma_0$ as a
substitution given by: $\sigma_0(x) = x$ and $\sigma_0(y) = \sigma(y)$ for any
$y \neq x$. Clearly, we have $\models_{\xM} \varphi\sigma_0$, as $\sigma_0$ and
$\sigma$ disagree only on the variable $x \notin X$ and $\Var(\varphi) \subseteq
X$.
Thus, by our condition, $x = x\sigma_0 = t\sigma_0$. Hence we know $t \in \xV$.
Similarly, one can take a substitution $\sigma_1$ as: $\sigma_1(x) = x'$ and
$\sigma_0(y) = \sigma(y)$ for any $y \neq x$, where $x'$ is a variable other
than $x$. Then, again by our condition, $x' = x\sigma_1 = t\sigma_1$.
Thus, we know from $x' \neq x$ that $t\sigma_1 \neq t\sigma_0$. From this, we
obtain $t = x$, as $\sigma_0$ and $\sigma_1$ disagree only on the term variable
$x$. Since $\xE \vdash \CEqn{X}{x}{x}{\varphi}$ by the \textit{Refl} rule, the
claim follows.

Now we proceed to the general case. By our assumption, we exclude the case where
$\varphi$ is not satisfiable. Let $\sigma$ be a $X$-valued substitution such
that $\models_{\xM} \varphi\sigma$.
Then, from the assumption we have $s\sigma = t\sigma$, and hence, we know that
$s,t$ are unifiable, and thus, there is a multi-hole context $C$ such that $s =
C[s_1,\ldots,s_n]$ and $t = C[t_1,\ldots,t_n]$ such that either one of $s_i$ or
$t_i$ is a variable for each $1 \leqslant i \leqslant n$. By
$C\sigma[s_1\sigma,\ldots,s_n\sigma] = s\sigma = t\sigma =
C\sigma[t_1\sigma,\ldots,t_n\sigma]$, we know $s_i\sigma = t_i\sigma$ for each
$1 \leqslant i \leqslant n$.
Thus, by our claim above for the case $s \in \xV$, and the \textit{Sym} rule, we
know that $\xE \vdash \CEqn{X}{s_i}{t_i}{\varphi}$ for each $1 \leqslant i
\leqslant n$. Then, using \textit{Refl}, \textit{Trans}, \textit{Cong} rules, it
is straightforward to obtain the desired $\xE \vdash \CEqn{X}{s}{t}{\varphi}$.
\end{proof}
}
\ifthenelse{\boolean{OmitProofs}}{}{\ProofAbstractingReflexivity}

\begin{restatable}{lemma}{LemmaAbstractingCalcStep}
\label{lem: abstracting calc step}
Let $\langle\xM, \xE \rangle$ be a CE-theory and
$\CEqn{X}{s}{t}{\varphi}$ a CE such that $\varphi$ is satisfiable.
Suppose $s\sigma \Cca^{=} t\sigma$ for all $X$-valued substitutions $\sigma$
such that $\Dom(\sigma) = X$ and $\models_{\xM} \varphi\sigma$. Then $\xE \vdash
\CEqn{X}{s}{t}{\varphi}$.
\end{restatable}

\begin{proof}[Proof (Sketch)]
First of all, the claim for the case $s,t \in \xT(\xFTh,X)$ follows using the
\textit{Axiom} rule. If $s\sigma = t\sigma$ for all $X$-valued substitutions
$\sigma$, the claim follows from Lemma~\ref{lem: abstracting reflexivity}. Thus,
it remains to consider the case that there exists an $X$-valued substitution
$\sigma$ such that $s\sigma \Cca  t\sigma$, $\models_{\xM} \varphi\sigma$ and $X
= \Dom(\sigma)$. Suppose that $s\sigma = C[f(v_1,\ldots,v_n)]_q$ and $t\sigma =
C[v_0]_q$ with $f \in \xFTh$ and $v_0,\ldots,v_n \in \Val$ such that
$\xI(f)(v_1,\ldots,v_n) = v_0$.
As $\Dom(\sigma) = X$ and $\sigma$ is $X$-valued, we have $s|_q =
f(s_1,\ldots,s_n)$ with $s_1,\ldots,s_n \in \Val \cup X$ and $t|_{q} \in \Val
\cup X$; hence $s|_q,t|_{q} \in \xT(\xFTh,X)$, and thus, $\xE \vdash
\CEqn{X}{s|_q}{t|_q}{\varphi}$ holds as we mentioned above.
Let $s = C_1[f(s_1,\ldots,s_n)]_q$ and $t = C_2[t_0]_q$ with
$C_1\sigma = C_2\sigma$ for any $X$-valued substitution $\sigma$ such that
$\Dom(\sigma) = X$ and $\models_{\xM} \varphi\sigma$. Then, $\xE \vdash
\CEqn{X}{s}{t}{\varphi}$ follows using $\xE \vdash
\CEqn{X}{s|_q}{t|_q}{\varphi}$, Lemma~\ref{lem: abstracting reflexivity}, and the
\textit{Cong} and \textit{Trans} rules.
\end{proof}

\newcommand{\ProofAbstractingCalcStep}{
\ifthenelse{\boolean{OmitProofs}}{\LemmaAbstractingCalcStep*}{}
\begin{proof}
We first show the case that $s,t \in \xT(\xFTh,X)$. Let $\sigma$ be any
$X$-valued substitution such that $\models_{\xM} \varphi\sigma$. Then,
$s\sigma,t\sigma \in \xT(\xFTh)$, and either $s\sigma = t\sigma$ or and $s\sigma
= C[f(v_1,\ldots,v_n)]$ and $t\sigma = C[v_0]$ (or vice versa) for some context
$C$ and $f \in \xFTh$ and $v_0,\ldots,v_n \in \Val$ such that
$\xI(f)(v_1,\ldots,v_n) = v_0$. Then, clearly, $\models_{\xM} s\sigma =
t\sigma$. Thus, $\models_{\xM} (\varphi\sigma \Rightarrow s\sigma = t\sigma)$.
Hence, using the \textit{Axiom} rule, it follows that $\xE \vdash
\CEqn{X}{s}{t}{\varphi}$.
If $s\sigma = t\sigma$ for all $X$-valued substitutions $\sigma$, then the claim
follows from Lemma~\ref{lem: abstracting reflexivity}. Thus, suppose there
exists an $X$-valued substitution $\sigma$ such that $s\sigma \Cca  t\sigma$,
$\models_{\xM} \varphi\sigma$ and $X = \Dom(\sigma)$. Then, we have $s\sigma =
C[f(v_1,\ldots,v_n)]_q$ and $t\sigma = C[v_0]_q$ for some context $C$ and $f \in
\xFTh$ and $v_0,\ldots,v_n \in \Val$ such that $\xI(f)(v_1,\ldots,v_n) = v_0$
(or vice versa).
Let $p \in \Pos(s)$ be the position such that $s|_p$ is the minimal subterm of
$s$ such that $f(v_1,\ldots,v_n) \unlhd s|_p\sigma$. Then, because $\Dom(\sigma)
= X$ and $\sigma$ is $X$-valued, we have $s|_p = f(s_1,\ldots,s_n)$ with
$s_1,\ldots,s_n \in \Val \cup X$. Hence, $p = q$.
Let $p' \in \Pos(t)$ be the position such that $t|_{p'}$ is the minimal subterm
of $t$ such that $v_0 \unlhd t|_{p'}\sigma$. Then, because $\Dom(\sigma) = X$
and $\sigma$ is $X$-valued, we have $t|_{p'} \in \Val \cup X$. Hence $p' = q$.
Thus, we have $p = p' = q$. Furthermore, $s|_q, t|_q \in \xT(\xFTh,X)$. Thus, it
follows from the previous arguments that $\xE \vdash
\CEqn{X}{s|_q}{t|_q}{\varphi}$.
Now, we have $s = C_1[f(s_1,\ldots,s_n)]_q$ and $t = C_2[t_0]_q$ for some
$s_1,\ldots,s_n,t_0 \in \Val \cup X$ and contexts $C_1,C_2$.
By our assumption, we have that $s\sigma =
C_1\sigma[f(s_1\sigma,\ldots,s_n\sigma)]_q \Cca^{=} C_2\sigma[t_0\sigma]_q$ for
all $X$-valued substitutions $\sigma$ such that $\Dom(\sigma) = X$ and
$\models_{\xM} \varphi\sigma$. Since $f(s_1\sigma,\ldots,s_n\sigma) = t_0\sigma$
never holds for such substitutions $\sigma$, we have that for each position $q'
\not\leqslant q$, $C_1\sigma|_{q'} = C_2\sigma|_{q'}$ holds for any $X$-valued
substitution $\sigma$ such that $\Dom(\sigma) = X$ and $\models_{\xM}
\varphi\sigma$. Thus, by Lemma~\ref{lem: abstracting reflexivity}, $\xE \vdash
\CEqn{X}{s|_{q'}}{t|_{q'}}{\varphi}$ holds for any $q' \in \Pos(s) \cap \Pos(t)$
such that $q' \not\leqslant q$.
Now, using \textit{Cong} and \textit{Trans} rules,
we have $\xE \vdash \CEqn{X}{s}{t}{\varphi}$.
\end{proof}
}
\ifthenelse{\boolean{OmitProofs}}{}{\ProofAbstractingCalcStep}

From Lemma~\ref{lem: abstracting calc step}, the partial completeness
for at most one calculation step follows.

\begin{theorem}
\label{thm: partial completeness of CEC0}
Let $\langle\xM, \xE \rangle$ be a CE-theory and
$\CEqn{X}{s}{t}{\varphi}$ a CE such that $\varphi$ is satisfiable.
Suppose $s\sigma \Cca^{=} t\sigma$ for all $X$-valued substitutions $\sigma$
such that $\models_{\xM} \varphi\sigma$. Then $\xE \vdash_{\mathbf{CEC}_0}
\CEqn{X}{s}{t}{\varphi}$.
\end{theorem}

One may expect that Theorem~\ref{thm: partial completeness of CEC0} can be
extended to the general completeness theorem for arbitrary $\xE$ in such a way
that $\xE \cec \CEqn{X}{s}{t}{\varphi}$ implies $\xE \vdash
\CEqn{X}{s}{t}{\varphi}$ (full completeness). Rephrasing this, we have: if
$s\sigma \Cab[\xE]{*} t\sigma$ for all $X$-valued substitutions $\sigma$ such
that $\models_{\xM} \varphi\sigma$, then $\xE \vdash \CEqn{X}{s}{t}{\varphi}$.
The partial completeness result above is far from this formulation of full
completeness in that the assumption does not assume arbitrary conversions
$s\sigma \Cab[\xE]{*} t\sigma$ but only $s\sigma \Cca^{=} t\sigma$ (i.e.\ at
most one calculation step). However, full completeness does not seem to hold for
the system $\mathbf{CEC}_0$, as witnessed by the following example.

\begin{example}
\label{exp:counter example for completeness}
Consider the following LCES:
\[
\xE = 
\left\{
\begin{array}{ll}
     \CEqn{\SET{ x }}{\mathsf{nneg}(x)}{\mathsf{true}}{x = 0} & (1)\\
     \CEqn{\SET{x,y}}{\mathsf{nneg}(x)}{\mathsf{nneg}(y)}{x + 1 = y} & (2)\\
\end{array}
\right.
\]
For each $n \ge 0$,
we have $\mathsf{nneg}(n) \Cab[\xE]{*} \mathsf{true}$:
\[
\mathsf{nneg}(n) \Cab[\xE]{} 
\mathsf{nneg}(n-1) \Cab[\xE]{}
\cdots \Cab[\xE]{}
\mathsf{nneg}(0) \Cab[\xE]{}
\mathsf{true}
\]
Thus, 
for the CE
$\CEqn{\SET{ x }}{\mathsf{nneg}(x)}{\mathsf{true}}{ x \ge 0}$,
we have
for all $\sigma$ such that $\models_{\xM} \sigma(x) \ge 0$,
$\mathsf{nneg}(x)\sigma \Cab[\xE]{*} \mathsf{true}\sigma = \mathsf{true}$.
On the other hand, for each $n \ge 0$,
one can give a derivation of
$\CEqn{\varnothing}{\mathsf{nneg}(n)}{\mathsf{true}}{}$---for example,
for $n = 2$,
\[
\infer[\it Trans]
  {\CEqn{\varnothing}{\mathsf{nneg}(2)}{\mathsf{true}}{}}
  {\infer[\it TInst]
     {\CEqn{\varnothing}{\mathsf{nneg}(2)}{\mathsf{nneg}(1)}{}}
     {\infer[\it Rule]
       {(2)}
       {}
     }
   &
   \infer[\it Trans]
    {\CEqn{\varnothing}{\mathsf{nneg}(1)}{\mathsf{true}}{}}
    {\infer[\it I]
       {\CEqn{\varnothing}{\mathsf{nneg}(1)}{\mathsf{nneg}(0)}{}}
       {\infer[\it Rule]
         {(2)}
         {}
       }
     &
     \infer[\it TInst]
       {\CEqn{\varnothing}{\mathsf{nneg}(0)}{\mathsf{true}}{}}
       {\infer[\it Rule]
         {(1)}
         {}
       }}}
\]
However, these derivations differ for each $n \ge 0$, and and are hardly
merged. As a conclusion, it seems that the CE $\CEqn{\SET{ x
}}{\mathsf{nneg}(x)}{\mathsf{true}}{ x \ge 0}$ is beyond the derivability of
$\mathbf{CEC}_0$.
\end{example}

\begin{remark}
After looking at Example \ref{exp:counter example for completeness}, it
might seem reasonable that adding some kind of induction reasoning is
required for our proof system. However, rules for the induction on positive
integers, etc.\ are only possible when working with a specific model. Such
rules have clearly a different nature than rules in our calculus
that work with any underlying model. Our calculus $\mathbf{CEC}_0$ intends
to be a general calculus that is free from specific underlying models and
does not include model-specific rules.
\end{remark}

\begin{remark}
We remark a difficulty to extend Theorem~\ref{thm: partial completeness
of CEC0} to multiple (calculation) steps, i.e.\ to have a statement like
$\xE \vdash_{\mathbf{CEC}_0} \CEqn{X}{s}{t}{\varphi}$ whenever $s\sigma \Cca^{*}
t\sigma$ for all $X$-valued substitutions $\sigma$ such that $\models_{\xM}
\varphi\sigma$. Or even to obtain a slightly weaker statement like
$\xE \vdash_{\mathbf{CEC}_0} \CEqn{X}{s}{t}{\varphi}$ whenever there exists a
natural number $n$ such that $s\sigma \Cca^{n} t\sigma$ for all $X$-valued
substitutions $\sigma$ such that $\models_{\xM} \varphi\sigma$. It might
look that induction on the length of $s\sigma \Cca^{*} t\sigma$ (or the one on
$n$ in the case of the latter) can be applied. However, to apply
Theorem~\ref{thm: partial completeness of CEC0} to each step, we need the form
$s_0\sigma \Cca s_1\sigma \Cca  \cdots \Cca  s_n\sigma$ for each $\sigma$ which
is not generally implied by $s\sigma \Cca^{*} t\sigma$ or $s\sigma \Cca^{n}
t\sigma$, as intermediate terms may vary depending on the substitution $\sigma$.
Extending Theorem~\ref{thm: partial completeness of CEC0} to multiple steps
remains as our future work.
\end{remark}

\begin{remark}
Our final remark deals with the difficulty to extend Theorem~\ref{thm: partial
completeness of CEC0} to a single rule step, i.e.\ to have a statement like $\xE
\vdash_{\mathbf{CEC}_0} \CEqn{X}{s}{t}{\varphi}$ whenever $s\sigma \Cru t\sigma$
for all $X$-valued substitutions $\sigma$ such that $\models_{\xM}
\varphi\sigma$. For the rule step, there may be multiple choices of positions
and multiple choices of CEs to be applied for the step $s\sigma \Cru t\sigma$.
Thus, we have to divide $X$-valued substitutions satisfying $\varphi$ depending
on each position $p$ that a CE is applied and each applied CE
$\CEqn{X_i}{\ell_i}{r_i}{\varphi_i} \in \xE$, and combine the obtained
consequences. However, it is in general not guaranteed that such a division of
substitutions can be characterized by a constraint. Note that the set of sets of
substitutions is in general not countable but the set of constraints is
countable. Thus, it may be necessary to assume some assumption on the
expressiveness of constraints to obtain the extension for the single rule step.
On the other hand, we conjecture that the (full) completeness would hold for
CE-theories with a finite underlying model.
\end{remark}

\section{Algebraic Semantics for CE-Validity}
\label{sec:algebraic-semantics}

In this section, we explore algebraic semantics for CE-validity. In this
approach, CE-validity is characterized by validity in models in a class of
algebras, which we call \emph{CE-algebras}. We show that this characterization is sound
and complete, in the sense that CE-validity can be fully characterized.

\subsection{CE-Algebras}

In this subsection, we introduce a notion of CE-algebras and validity in them.
After presenting basic properties of our semantics, we prove its soundness with
respect to the CE-validity.

\begin{definition}[CE-$\langle \Sigma,\xM \rangle$-algebra]
Let $\Sigma = \langle\xSTh,\xSTe, \xFTh, \xFTe \rangle$ be a signature and $\xM
= \langle \xI, \xJ \rangle$ be a model over $\xSTh$ and $\xFTh$. A \emph{constrained
equational $\langle \Sigma,\xM \rangle$-algebra} (CE-$\langle \Sigma,\xM
\rangle$-algebra, for short) is a pair $\fM = \langle \fI, \fJ \rangle$ where
$\fI$ assigns each $\tau \in \xS$ a non-empty set $\fI(\tau)$, specifying
its domain, and $\fJ$ assigns each $f\colon \tau_1 \times \cdots \times \tau_n
\to \tau_0 \in \xF$ an interpretation function $\fJ(f)\colon \fI(\tau_1) \times
\cdots \times \fI(\tau_n) \to \fI(\tau_0)$ that extends the model $\xM = \langle
\xI, \xJ \rangle$, that is, $\fI(\tau) \supseteq \xI(\tau)$ for all $\tau \in
\xSTh$ and $\fJ(f)\restriction_{\xI(\tau_1) \times \cdots \times \xI(\tau_n)} =
\xJ(f)$ for all $f \in \xFTh$ (or more generally there exists an injective
homomorphism $\iota\colon \xM \to \fM$).
\end{definition}

Let $\fM = \langle \fI, \fJ \rangle$ be a CE-$\langle \Sigma,\xM
\rangle$-algebra. A valuation $\rho$ over $\fM$ is defined similarly to $\xM$,
but $\xS$ instead of $\xSTh$, $\xF$ instead of $\xFTh$, etc.
Similarly, a valuation $\rho$ over $\fM$ satisfies a logical constraint
$\varphi$, denoted by $\models_{\fM} \varphi$, if $\inter{\varphi}_{\rho,\fM} =
\mathsf{true}$.

Careful readers may wonder why the interpretation functions for the theory part
are not the same but an extension of the underlying model $\xM = \langle
\xI, \xJ \rangle$. Indeed, in the definition of CE-$\langle \Sigma,\xM
\rangle$-algebras $\fM = \langle \fI, \fJ \rangle$ above, we only demand that
$\fI(\tau) \supseteq \xI(\tau)$ for all $\tau \in \xSTh$ and not $\fI(\tau) =
\xI(\tau)$ for all $\tau \in \xSTh$.
In fact, this is required to obtain the completeness result; however,
this explanation is postponed until Example~\ref{exp:why underlying model
is an extension}. We continue to present some basic properties of our
semantics which are proven in a straightforward manner.

\begin{lemma}
\label{lem:interpretation identity is closed under substitutions and contexts}
    Let $\fT = \langle \xM, \xE \rangle$ 
    be a CE-theory over a signature $\Sigma$, and
    $\fM = \langle \fI, \fJ \rangle$ a CE-$\langle \Sigma,\xM\rangle$-algebra.
    Then, the binary relation over terms given by $\inter{\cdot}_{\fM,\rho} = \inter{\cdot}_{\fM,\rho}$ 
    for any valuation $\rho$ on $\fM$, 
    is closed under substitutions and contexts.
\end{lemma}

\begin{lemma}
\label{lem:preservation_of_validity_in_underlying_model} 
    Let $\fT = \langle \xM, \xE \rangle$ 
    be a CE- theory over a signature $\Sigma$ such that $\xM = \langle \xI, \xJ \rangle$,
    $\fM$ a CE-$\langle \Sigma,\xM\rangle$-algebra, and
    $X \subseteq \xVTh$ a set of theory variables and suppose
    $t \in \xT(\xFTh,X)$. 
    Then, for any valuation $\rho$ on $\fM$ such that $\rho(x) \in \xI(\tau)$ for all $x^\tau \in X$,
    we have $\inter{t}_{\fM,\rho} = \inter{t}_{\xM,\rho}$.
\end{lemma}

Next, we extend the definition of validity on CE-algebras
to CEs,
by which we can give a notion of models of CE-theories,
and the semantic consequence relation.

\begin{definition}[model of constrained equational theory]
Let $\fT = \langle \xM, \xE \rangle$ be a CE-theory
over a signature $\Sigma$ such that $\xM = \langle \xI, \xJ \rangle$,
and
$\fM = \langle \fI, \fJ \rangle$ a CE-$\langle \Sigma,\xM \rangle$-algebra.
\begin{enumerate}
    \item 
A CE
$\CEqn{X}{\ell}{r}{\varphi}$
is said to be \emph{valid} in $\fM$,
denoted by $\models_{\fM} \CEqn{X}{\ell}{r}{\varphi}$,
if for all valuations $\rho$ over $\fM$ 
satisfying the constraint $\varphi$ (i.e.\ $\inter{\varphi}_{\fM,\rho} = \m{true}$ holds)
and $\rho(x) \in \xI(\tau)$ holds for all $x^\tau \in X$,
we have $\inter{\ell}_{\fM,\rho} = \inter{r}_{\fM,\rho}$.
\item
    A CE-$\langle \Sigma,\xM\rangle$-algebra $\fM = \langle \fI, \fJ \rangle$
    is said to be a \emph{model} of the CE-theory $\fT$
    if $\models_{\fM} \xE$. 
    Here, $\models_{\fM} \xE$ denotes that $\models_{\fM} \CEqn{X}{\ell}{r}{\varphi}$
    for all $\CEqn{X}{\ell}{r}{\varphi} \in \xE$.
\item
    Let $\CEqn{X}{\ell}{r}{\varphi}$ be a CE.
    We write $\fT \models \CEqn{X}{\ell}{r}{\varphi}$
    (or $\xE \models \CEqn{X}{\ell}{r}{\varphi}$ if no confusion arises)
    if $\models_{\fM} \CEqn{X}{\ell}{r}{\varphi}$ holds for all CE-$\langle \Sigma,\xM\rangle$-algebras
    $\fM$ that are models of $\fT$.
\end{enumerate}
\end{definition}
We remark that, in item \Bfnum{1}, as $\varphi \in \xT(\xFTh,\xVTh)$, we have
$\inter{\varphi}_{\fM,\rho} = \m{true}$ if and only if
$\inter{\varphi}_{\xM,\rho} = \m{true}$ by
Lemma~\ref{lem:preservation_of_validity_in_underlying_model}. Based on the
preceding lemmas,  soundness of our semantics with respect to conversion is not
difficult to obtain.

\begin{restatable}[soundness w.r.t.\ conversion]{lemma}{LemmaSoundnessWrtConversion}
\label{lem:soundness w.r.t. conversion}
    Let $\fT = \langle \xM, \xE \rangle$ 
    be a CE-theory over a signature $\Sigma$, and
    $\fM = \langle \fI, \fJ \rangle$ a CE-$\langle \Sigma,\xM\rangle$-algebra
    such that $\models_{\fM} \xE$.
    If $s \Cab[\xE]{*} t$ then $\inter{s}_{\fM,\rho} = \inter{t}_{\fM,\rho}$
    for any valuation $\rho$ on $\fM$.
\end{restatable}

\begin{proof}[Proof (Sketch)]
It suffices to consider the case $s \Cab[\xE]{} t$ with a root step;
the claim easily follows from Lemma~\ref{lem:interpretation identity
is closed under substitutions and contexts}. Let $\Sigma = \langle\xSTh,\xSTe,
\xFTh, \xFTe \rangle$ and $\xM = \langle \xI, \xJ \rangle$.
Let $s \Cca t$. Then, $s,t \in \xT(\xFTh)$, and hence $\inter{s}_{\xM} = \inter{t}_{\xM}$.
Thus, $\inter{s}_{\fM} = \inter{t}_{\fM}$ by Lemma~\ref{lem:preservation_of_validity_in_underlying_model}.
Otherwise, let $s \Cru[\xE] t$.
Then, $s = \ell\sigma$ and $t = r\sigma$ for some $\CEqn{X}{\ell}{r}{\varphi} \in \xE$
and an $X$-valued substitution $\sigma$ such that $\models_{\xM} \varphi\sigma$.
We have a valuation $\inter{\sigma}_{\fM,\rho}$ on $\fM$
by $\inter{\sigma}_{\fM,\rho}(y) = \inter{\sigma(y)}_{\fM,\rho}$ for any $y \in \xV$.
Then, similarly to Lemma~\ref{lem:substitution as interpretation},
we have $\inter{u\sigma}_{\fM,\rho} = \inter{u}_{\fM, \inter{\sigma}_{\fM,\rho}}$
for any term $u \in \xT(\Sigma,\xV)$.
Furthermore, for $x \in X$,
$\inter{\sigma}_{\fM,\rho}(x) = \inter{\sigma(x)}_{\fM,\rho} = \sigma(x)$ holds.
Hence, by Lemma~\ref{lem:preservation_of_validity_in_underlying_model},
$\inter{\varphi}_{\fM,\inter{\sigma}_{\fM,\rho}} = \mathsf{true}$.
Thus, 
$\inter{s}_{\fM,\rho} = \inter{\ell}_{\fM, \inter{\sigma}_{\fM,\rho}}
= \inter{r}_{\fM, \inter{\sigma}_{\fM,\rho}}
= \inter{t}_{\fM,\rho}$.
\end{proof}

\newcommand{\ProofSoundnessWrtConversion}{
\ifthenelse{\boolean{OmitProofs}}{\LemmaSoundnessWrtConversion*}{}
\begin{proof}
Clearly, it suffices to show that $s \Cab[\xE]{} t$
implies $\inter{s}_{\fM,\rho} = \inter{t}_{\fM,\rho}$ for any valuation $\rho$ on $\fM$.
Moreover, by Lemma~\ref{lem:interpretation identity is closed under substitutions and contexts},
it suffices to consider the case where $s \Cab[\xE]{} t$ is a root step.
Let $\Sigma = \langle\xSTh,\xSTe, \xFTh, \xFTe \rangle$ and $\xM = \langle \xI, \xJ \rangle$.
First, we consider the case $s \Cca t$.
Then, $s = f(s_1,\ldots,s_n)$, $s_1,\ldots,s_n,t \in \Val$, $f \in \xFTh$,
and $\xI(f)(s_1,\ldots,s_n) = t$ (or vice versa).
Thus, $s,t \in \xT(\xFTh)$,
and we have $\inter{s}_{\xM} = \inter{t}_{\xM}$,
and hence $\inter{s}_{\fM} = \inter{t}_{\fM}$ follows
by Lemma~\ref{lem:preservation_of_validity_in_underlying_model}.
By $s,t \in \xT(\xFTh)$,
we have $\inter{s}_{\fM,\rho} = \inter{t}_{\fM,\rho}$ for any valuation $\rho$.
Next, we consider the case $s \Cru[\xE] t$.
Then, $s = \ell\sigma$ and $t = r\sigma$ for some $\CEqn{X}{\ell}{r}{\varphi} \in \xE$
and an ($X$-valued) substitution $\sigma$ 
such that $\sigma(x) \in \xI(\tau)$ for $x^\tau \in X$
and $\models_{\xM} \varphi\sigma$.
Let $\xi$ be a valuation on $\fM$ such that $\xi(x) = \sigma(x)$ for each $x \in X$.
Then, as $\varphi \in \xT(\xFTh,X)$, by Lemma~\ref{lem:preservation_of_validity_in_underlying_model},
we obtain $\inter{\varphi}_{\fM,\xi} = \inter{\varphi}_{\xM,\sigma} = \mathsf{true}$.
Thus, by our assumption that $\models_{\fM} \xE$,
it follows that $\inter{\ell}_{\fM,\xi} = \inter{r}_{\fM,\xi}$.
Thus, we have $\inter{\ell}_{\fM,\xi} = \inter{r}_{\fM,\xi}$
for any valuation $\xi$ such that $\xi(x) = \sigma(x)$ for each $x \in X$.
Now, define a valuation $\inter{\sigma}_{\fM,\rho}$ on $\fM$
by $\inter{\sigma}_{\fM,\rho}(y) = \inter{\sigma(y)}_{\fM,\rho}$ for any $y \in \xV$.
Then, it can be shown by induction on $u$
that $\inter{u\sigma}_{\fM,\rho} = \inter{u}_{\fM, \inter{\sigma}_{\fM,\rho}}$
for any term $u \in \xT(\Sigma,\xV)$,
similarly to Lemma~\ref{lem:substitution as interpretation}.
Furthermore, for $x \in X$,
as $\sigma(x) \in \xI(\tau)$,
$\inter{\sigma}_{\fM,\rho}(x) = \inter{\sigma(x)}_{\fM,\rho} = \sigma(x)$ holds.
Thus, 
$\inter{s}_{\fM,\rho} = \inter{\ell}_{\fM, \inter{\sigma}_{\fM,\rho}}
= \inter{r}_{\fM, \inter{\sigma}_{\fM,\rho}}
= \inter{t}_{\fM,\rho}$.
\end{proof}
}
\ifthenelse{\boolean{OmitProofs}}{}{\ProofSoundnessWrtConversion}

Now we present the soundness of our semantics with respect to the CE-validity.

\begin{restatable}[soundness w.r.t.\ CE-validity]{theorem}{TheoremSoundnessOfAlgebraicSemantics}
\label{thm:soundness of algebraic semantics w.r.t. constrained equational validity}
    Let $\fT$ be a CE-theory.
    If $\fT \cec \CEqn{X}{s}{t}{\varphi}$,
    then $\fT \models \CEqn{X}{s}{t}{\varphi}$.
\end{restatable}

\begin{proof}[Proof (Sketch)]
Let $\fT = \langle \xM, \xE \rangle$ and $\xM = \langle \xI, \xJ \rangle$.
Suppose $\fM = \langle \fI, \fJ \rangle$ is a CE-$\langle \Sigma,\xM\rangle$-algebra 
such that $\models_{\fM} \xE$.
Let $\rho$ be a valuation over $\fM$ 
satisfying the constraints $\varphi$
and $\rho(x) \in \xI(\tau)$ holds for all $x^\tau \in X$.
Now, let $\hat \rho$ be a valuation that is obtained
from $\rho$ by restricting its domain to $X$.
Then, $\models_{\xM} \varphi\hat\rho$ by Lemma~\ref{lem:preservation_of_validity_in_underlying_model},
and thus $s\hat\rho \Cab[\xE]{*} t\hat\rho$ holds.
Hence, by~Lemma~\ref{lem:soundness w.r.t. conversion},
$\inter{s\hat\rho}_{\fM,\tau} = \inter{t\hat\rho}_{\fM,\tau}$ holds
for any valuation $\tau$.
This means that $\inter{s}_{\fM,\tau'} = \inter{t}_{\fM,\tau'}$
for any extension $\tau'$ of $\hat\rho$. In particular, 
one obtains  $\inter{s}_{\fM,\rho} = \inter{t}_{\fM,\rho}$.
\end{proof}

\newcommand{\ProofSoundnessOfAlgebraicSemantics}{
\ifthenelse{\boolean{OmitProofs}}{\TheoremSoundnessOfAlgebraicSemantics*}{}
\begin{proof}
  Suppose that $\fT = \langle \xM, \xE \rangle$ is a CE-theory over 
  a signature $\Sigma$.
  Let $\xM = \langle \xI, \xJ \rangle$.
  Suppose moreover $\xE \cec \CEqn{X}{s}{t}{\varphi}$.
  We show that for any CE-$\langle \Sigma,\xM\rangle$-algebra $\fM = \langle \fI, \fJ \rangle$ 
  such that $\models_{\fM} \xE$, we have $\models_{\fM} \CEqn{X}{s}{t}{\varphi}$.
  Suppose $\fM = \langle \fI, \fJ \rangle$ is a CE-$\langle \Sigma,\xM\rangle$-algebra 
  such that $\models_{\fM} \xE$.
  Then, by the definition of CE-validity,
  $s\sigma \Cab[\xE]{*} t\sigma$ holds for any 
  $X$-valued substitution $\sigma$ such that $\models_{\xM} \varphi\sigma$.
Let $\rho$ be a valuation over $\fM$ 
satisfying the constraints $\varphi$
and $\rho(x) \in \xI(\tau)$ holds for all $x^\tau \in X$.
Now, let $\hat \rho$ be a valuation that is obtained
from $\rho$ by restricting its domain to $X$.
Then since $\rho(x) \in \xI(\tau) \cong \Val^\tau$ for all $x \in X$, 
$\hat \rho$ is a substitution such that $\Dom(\hat \rho) = \VDom(\hat \rho) = X$.
Furthermore, as $\Var(\varphi) \subseteq X$, 
we have by Lemma~\ref{lem:preservation_of_validity_in_underlying_model} 
that $\models_{\xM} \varphi\hat\rho$.
Thus, $\hat \rho$ is a $X$-valued substitution satisfying $\models_{\xM} \varphi\hat\rho$,
and thus, by our assumption, 
$s\hat\rho \Cab[\xE]{*} t\hat\rho$ holds.
Thus, by~Lemma~\ref{lem:soundness w.r.t. conversion},
$\inter{s\hat\rho}_{\fM,\tau} = \inter{t\hat\rho}_{\fM,\tau}$ holds
for any valuation $\tau$.
This means that for any extension $\tau'$ of $\hat\rho$,
$\inter{s}_{\fM,\tau'} = \inter{t}_{\fM,\tau'}$.
In particular, one can take $\tau' := \rho$.
Then $\inter{s}_{\fM,\rho} = \inter{t}_{\fM,\rho}$ holds.
Thus we have now shown that $\inter{s}_{\fM,\rho} = \inter{t}_{\fM,\rho}$ holds,
for all valuations $\rho$ over $\fM$ 
satisfying the constraints $\varphi$
and $\rho(x) \in \xI(\tau)$ holds for all $x^\tau \in X$.
Therefore, we conclude $\models_{\fM} \CEqn{X}{s}{t}{\varphi}$.
\end{proof}
}
\ifthenelse{\boolean{OmitProofs}}{}{\ProofSoundnessOfAlgebraicSemantics}

The combination of Theorem~\ref{thm:soundness of the system CEC0} and
Theorem~\ref{thm:soundness of algebraic semantics w.r.t. constrained equational
validity} implies the following corollary.

\begin{corollary}[soundness of $\mathbf{CEC}_0$ w.r.t.\ algebraic semantics]
\label{cor:soundness of CEC0 wrt algebraic semantics}
    Let $\fT$ be a CE-theory.
    If $\fT \vdash_{\mathbf{CEC}_0} \CEqn{X}{s}{t}{\varphi}$,
    then $\fT \models \CEqn{X}{s}{t}{\varphi}$.
\end{corollary}

\begin{example}
Consider integer arithmetic for the underlying model $\xM$. Take a term
signature $\xFTe = \SET{ \m{a}\colon \textsf{Int}}$. Consider the LCES
$\xE= \SET{ \m{a} \approx \m{0}, \m{a} \approx \m{1} }$ with $\m{0},\m{1} \in \Val$ and
$0,1 \in \mathbb{Z}$, hence $\xJ(\m{0}) = 0$ and $\xJ(\m{1}) = 1$.
Then, for any valuation $\rho$ on a CE-$\langle \Sigma,\xM\rangle$-algebra $\fM
= \langle \fI, \fJ \rangle$ we have $\rho(\m{0}) = \m{0}$ and $\rho(\m{1}) =
\m{1}$. Thus, if $\fM$ is a model of $\xE$ then it follows that $0 =
\inter{\m{0}} = \inter{\m{a}} = \inter{\m{1}} = 1$, which is a contradiction.
Therefore, there is no CE-$\langle \Sigma,\xM\rangle$-algebra $\fM$ which
validates $\xE$.
\end{example}

This example motivates us to introduce the following definition of consistency
for CE-theories.

\begin{definition}[consistency]
\label{def:consistency}
A CE-theory is said to be \emph{consistent} if it has a model.
\end{definition} 
Our definition of consistent CE-theories does not exclude any
theory that has only an almost trivial model such that $\xI(\tau) = \SET{
\bullet }$ for all $\tau \in \xSTe$.

\subsection{Completeness w.r.t.\ CE-Validity}

In this subsection, we prove the completeness of algebraic semantics with
respect to the CE-validity. That is, if a CE is valid in all models of a
CE-theory then it is a CE-consequence of the CE-theory. We start by defining
congruence relations, quotient algebras and term algebras that suit our
formalism, incorporating standard notions for example the first-order equational
logic, and then present basic results on them.

Let $\Sigma = \langle\xSTh,\xSTe, \xFTh, \xFTe \rangle$ be a signature,
$\xM = \langle \xI, \xJ \rangle$ a model over $\xSTh$ and $\xFTh$,
and $\fM = \langle \fI, \fJ \rangle$ a CE-$\langle \Sigma,\xM\rangle$-algebra.
A \emph{congruence relation} on $\fM$ is an 
$\xS$-indexed family of relations 
${\sim} = ({\sim}^\tau)_{\tau \in \xS}$
that satisfies all of the following:
\begin{enumerate}
\item $\sim^\tau$ is an equivalence relation on $\fI(\tau)$,
\item $\sim^\tau \cap~ \xI(\tau)^2$ is the identity relation for $\tau \in \xSTh$, and
\item for each $f\colon \tau_1 \times \cdots \times \tau_n \to \tau_0 \in \xF$,
if $a_i \sim^{\tau_i} b_i$ for all $1 \leqslant i \leqslant n$
then $\fJ(f)(a_1,\ldots,a_n) \sim^{\tau_0} \fJ(f)(b_1,\ldots,b_n)$.
\end{enumerate}
We note here that the difference from the standard notion of
congruence relation on algebras is located in item \Bfnum{2}.
Given a CE-$\langle \Sigma,\xM\rangle$-algebra
$\fM = \langle \fI, \fJ \rangle$ and a congruence relation $\sim$
on it, the quotient CE-$\langle \Sigma,\xM\rangle$-algebra
$\fM/{\sim} = \langle \fI', \fJ' \rangle$
is defined as follows:
$\fI'(\tau) = \fI(\tau)/{\sim^\tau}
= \SET{ {}[a]_{\sim^\tau} \mid a \in \fI(\tau) }$
where $[a]_{\sim^\tau}$ is the equivalence class of $a \in \fI(\tau)$,
i.e.\ $[a]_{\sim^\tau} = \SET{ b \in \fI(\tau) \mid a \sim^\tau b }$,
and 
$\fJ'(f)([a_1]_{\sim^{\tau_1}},\ldots,[a_n]_{\sim^{\tau_n}})
= [\fJ(f)(a_1,\ldots,a_n)]_{\sim^{\tau_0}}$.
It is easy to see 
that $\fJ'$ is well-defined provided that $\sim$ is a congruence.
When no confusion occurs, we omit the superscript $\tau$ from $\sim^\tau$.

\begin{restatable}[quotient algebra]{lemma}{LemmaQuotientAlgebra}
\label{lem:quotient algebra}
Let $\fM$ be a CE-$\langle \Sigma,\xM\rangle$-algebra,
and $\sim$ a congruence on it.
Then
$\fM/{\sim}$ is a CE-$\langle \Sigma,\xM\rangle$-algebra.
\end{restatable}

\newcommand{\ProofQuotientAlgebra}{
\ifthenelse{\boolean{OmitProofs}}{\LemmaQuotientAlgebra*}{}
\begin{proof}
Suppose $\Sigma = \langle\xSTh,\xSTe, \xFTh, \xFTe \rangle$ 
and $\xM = \langle \xI, \xJ \rangle$.
Let $\fM/{\sim} = \langle \fI', \fJ' \rangle$.
Then $\fI', \fJ'$ are well-defined as we mentioned above.
By the definition of ${\sim} = ({\sim}^\tau)_{\tau \in \xS}$,
the domain and codomain of $\fJ'(f)$ respect the sort of function 
symbols $f$.
Lastly, 
since $\sim^\tau \cap ~ \xI(\tau)^2$ is the identify relation for $\tau \in \xSTh$,
we obtain that $\fM/{\sim}$ restricted to
$(\xI(\tau))_{\tau \in \xSTh}$
and $\xFTh$ is isomorphic to $\xM$.
Therefore, $\fM/{\sim}$ is again a CE-$\langle \Sigma,\xM\rangle$-algebra.
\end{proof}
}
\ifthenelse{\boolean{OmitProofs}}{}{\ProofQuotientAlgebra}

Next, we define the term algebra. In contrast to the usual
construction, for term CE-algebras we need to take care of
identification induced by underlying models.

\begin{definition}[term algebra]
    Let $\Sigma = \langle\xSTh,\xSTe, \xFTh, \xFTe \rangle$ be a signature,
    $\xM = \langle \xI, \xJ \rangle$ a model over $\xSTh$ and $\xFTh$,
    and $U$ a set of variables.
    The \emph{term algebra generated from $U$ with $\xM$} 
        (denoted by $T[\xM](\Sigma,U)$) 
     is a pair
    $\fM = \langle \fI, \fJ \rangle$ where
    \begin{itemize}
    \item $\fI(\tau) =  \xT(\xF, U)^\tau/{\sim_\mathrm{c}}$, and
    \item $\fJ(f)([t_1]_\mathrm{c},\ldots,[t_n]_\mathrm{c}) = [f(t_1,\ldots,t_n)]_\mathrm{c}$
    for any $f \in \xF$.
    \end{itemize}
    Here, $\xF = \xFTh \cup \xFTe$, ${\sim_\mathrm{c}} = {\Cca^{*}}$,
    and $[t]_\mathrm{c}$ denotes the $\sim_\mathrm{c}$-equivalence class containing a term $t$.
    Since ${\Cca^{*}}$ is sort preserving, we regard $\sim_\mathrm{c}$
    as the sum of the $\tau$-indexed family of relations $\sim_\mathrm{c}^\tau$ with $\tau \in \xS$.
    Clearly, $\fJ(f)$ is well-defined, since ${\Cca^{*}}$ is closed under contexts.
\end{definition}

\begin{restatable}{lemma}{LemmaTermAlgebra}
\label{lem:term algebra}
    Let $\Sigma = \langle\xSTh,\xSTe, \xFTh, \xFTe \rangle$ be a signature,
    $\xM$ a model over $\xSTh$ and $\xFTh$,
    and $U$ a set of variables.
    Then, the term algebra $T[\xM](\Sigma,U)$
    is a CE-$\langle \Sigma,\xM\rangle$-algebra.
\end{restatable}

\newcommand{\ProofTermAlgebra}{
\ifthenelse{\boolean{OmitProofs}}{\LemmaTermAlgebra*}{}
\begin{proof}
Let $\xM = \langle \xI, \xJ \rangle$ and
$T[\xM](\Sigma,U) = \langle \fI, \fJ \rangle$.
First, we show $\xI(\tau) \subseteq \fI(\tau)$ for any $\tau \in \xSTh$.
Let $u,v \in \Val^\tau$.
Then $u \sim_\mathrm{c} v$ implies $u = v$ by Lemma~\ref{lem:properties of calculation steps}.
Hence for any $u,v \in \Val^\tau$, $[u]_\mathrm{c} = [v]_\mathrm{c}$ if and only if $u = v$.
Thus, we have $\xI(\tau) = \Val^\tau 
\cong \SET{ {}[v]_\mathrm{c} \mid v \in \Val^\tau }
\subseteq \xT(\Sigma,U)^\tau/{\sim_\mathrm{c}}$.
Here, $A \cong B$ denotes that sets $A,B$ are isomorphic
(and isomorphic sets are identified).
Thus, $\xI(\tau) \subseteq \fI(\tau)$ (up to isomorphism).
Next, let $v_1,\ldots,v_n \in \Val$
and $\xJ(f)(v_1,\ldots,v_n) = v_0$.
Then, by $f(v_1,\ldots,v_n) \Cca^* v_0$,
we have $f(v_1,\ldots,v_n) \sim_\mathrm{c} v_0$.
Therefore, $\fJ(f)([v_1]_\mathrm{c},\ldots,[v_n]_\mathrm{c}) 
= [f(v_1,\ldots,v_n)]_\mathrm{c}
= [v_0]_\mathrm{c} = v_0 = \xJ(f)(v_1,\ldots,v_n)$.
Note that $v_0$ and $[v_0]_\mathrm{c}$ are identified
based on $\Val^\tau \cong \SET{ {}[v]_\mathrm{c} \mid v \in \Val^\tau }$.
\end{proof}
}
\ifthenelse{\boolean{OmitProofs}}{}{\ProofTermAlgebra}

We introduce a syntactic counter part of the notion of consistency of
CE-theories for which equivalence of the two notions will be proved
only briefly.

\begin{definition}[consistency w.r.t.\ values]
\label{def:consistency w.r.t. values}
    A CE-theory $\fT = \langle \xM, \xE \rangle$ is said to be \emph{consistent with respect to values} (value-consistent, for short)
    if for any $u,v\in \Val^\tau$,
    $u \Cab[\xE]{*} v$ implies $u = v$.
\end{definition}

Based on the preparations so far, we now proceed to construct canonical models
of CE-theories. The first step is to show that $\Cab[\xE]{*}$ is a congruence
relation  on the term algebra for each CE-theory $\fT = \langle \xM, \xE
\rangle$; special attention on $\sim_\mathrm{c}$ is required.

\begin{restatable}{lemma}{LemmaCongruenceOnTermAlgebra}
\label{lem:congruence on term algebra}
    Let $\fT = \langle \xM, \xE \rangle$
    be a value-consistent CE-theory
    over a signature $\Sigma$,
    and
    $U$ a set of variables.
    For any $[s]_{\mathrm{c}},[t]_{\mathrm{c}} \in T[\xM](\Sigma,U)$,
    let ${\sim}_{\xE} = \SET{ \langle [s]_{\mathrm{c}}, [t]_{\mathrm{c}} \rangle \mid s \Cab[\xE]{*} t }$.
    Then, $\sim_\xE$ is a congruence relation 
    on the term algebra $T[\xM](\Sigma,U)$.
\end{restatable}

\begin{proof}[Proof (Sketch)]
Note first that $\sim_\xE$ is well-defined because one has always ${\Cca^{*}} \subseteq {\Cab[\xE]{*}}$.
    Let $\Sigma = \langle\xSTh,\xSTe, \xFTh, \xFTe \rangle$, $\xM = \langle \xI, \xJ \rangle$,
    and $T[\xM](\Sigma,U) = \langle \fI, \fJ \rangle$.
    We only present a proof that ${\sim}_\xE^\tau \cap {\xI(\tau)}^2$ 
    equals the identity relation for $\tau \in \xSTh$ here.
    Let $\tau \in \xSTh$
    and suppose $[u]_\mathrm{c} \sim_\xE^\tau [v]_\mathrm{c}$ with $u, v \in \xI(\tau) \cong \Val^\tau$.
    Then, we have $u \Cab[\xE]{*} v$ by the definition of $\sim_\xE$, and by consistency w.r.t.\ values of the theory $\fT$, 
    we obtain $u = v$ as $u,v \in \Val$. Therefore, $[u]_\mathrm{c} = [v]_\mathrm{c}$.
\end{proof}

\newcommand{\ProofCongruenceOnTermAlgebra}{
\ifthenelse{\boolean{OmitProofs}}{\LemmaCongruenceOnTermAlgebra*}{}
\begin{proof}
    Note first that $\sim_\xE$ is well-defined because one has always ${\Cca^{*}} \subseteq {\Cab[\xE]{*}}$.
    Let $\Sigma = \langle\xSTh,\xSTe, \xFTh, \xFTe \rangle$, $\xM = \langle \xI, \xJ \rangle$,
    and $T[\xM](\Sigma,U) = \langle \fI, \fJ \rangle$.
    Since $s \Cab[\xE]{*} t$ implies terms $s,t$ have the same sort, 
    and thus one can regard $\sim_\xE$ as the sum of $\tau$-indexed $\sim_\xE^\tau$
    with $\tau \in \xS$.
    First, since $\Cab[\xE]{*}$ is an equivalence relation,
    ${\sim}_\xE$ is an equivalence relation.
    Next, we show that ${\sim}_\xE^\tau \cap {\xI(\tau)}^2$ 
    equals the identity relation for $\tau \in \xSTh$.
    For this, let $\tau \in \xSTh$
    and suppose $[u]_\mathrm{c} \sim_\xE^\tau [v]_\mathrm{c}$ with $u, v \in \xI(\tau) \cong \Val^\tau$.
    Then, we have $u \Cab[\xE]{*} v$ by the definition of $\sim_\xE$, and by consistency with respect to values of the theory $\fT$, 
    we obtain $u = v$ as $u,v \in \Val$. Hence, $[u]_\mathrm{c} = [v]_\mathrm{c}$.
    Finally, assume $f \in \xF$ and $[s_i]_\mathrm{c} \sim_\xE [t_i]_\mathrm{c}$  for $1 \leqslant i \leqslant n$.
    Then $s_i \Cab[\xE]{*} t_i$ for $1 \leqslant i \leqslant n$ by the definition of $\sim_\xE$. 
    Hence $f(s_1,\ldots,s_n) \Cab[\xE]{*} f(t_1,\ldots,t_n)$.
    Thus, $\fJ(f)([s_1]_\mathrm{c},\ldots,[s_n]_\mathrm{c}) = [f(s_1,\ldots,s_n)]_\mathrm{c}
    \sim_\xE [f(t_1,\ldots,t_n)]_\mathrm{c} = \fJ(f)([t_1]_\mathrm{c},\ldots,[t_n]_\mathrm{c})$.
\end{proof}
}
\ifthenelse{\boolean{OmitProofs}}{}{\ProofCongruenceOnTermAlgebra}

We give a construction of canonical models for each CE-theory $\xT$.

\begin{restatable}{lemma}{LemmaValidityOnTermAlgebra}
\label{lem:validity on term algebra}
    Let $\fT = \langle \xM, \xE \rangle$
    be a value-consistent CE-theory
    over a signature $\Sigma$.
    Then, 
    the quotient $\xT_\xE = T[\xM](\Sigma,\xV)/{\sim_\xE}$ of the term algebra 
    is a CE-$\langle \Sigma,\xM\rangle$-algebra.
    Furthermore, both of the following hold:
    \begin{enumerate}
    \item $\models_{\xT_\xE} \CEqn{X}{s}{t}{\varphi}$
    if and only if 
    $\xE \cec \CEqn{X}{s}{t}{\varphi}$,
    and 
    \item $\models_{\xT_\xE} \xE$.
    \end{enumerate}
\end{restatable}

\begin{proof}[Proof (Sketch)]
That $\xT_\xE$ is a CE-$\langle \Sigma,\xM\rangle$-algebra
    follows from Lemmas~\ref{lem:quotient algebra} and~\ref{lem:congruence on term algebra}.
    Let us abbreviate $[[t]_\mathrm{c}]_{{\sim}_\xE}$ as $[t]_\xE$.
    First we claim that $[u\sigma]_\xE = \inter{u}_{\xT_\xE,\rho}$ holds for any term $u$,
    for any substitution $\sigma$ and valuation $\rho$ on $\xT_\xE$
    such that $\rho(x) = [\sigma(x)]_\xE$, using induction on $u$. 
    \Bfnum{1.}
   ($\Rightarrow$)
    Let $\sigma$ be an $X$-valued substitution such that $\models_\xM \varphi\sigma$.
    Take a valuation $\rho$ on $\xT_\xE$ as $\rho(x) = [\sigma(x)]_\xE$.
    Then, $\rho(x) \in \xI(\tau)$ for all $x^\tau \in X$ and $\models_\xM \varphi\rho$.
    Thus, $\inter{s}_{\xT_\xE,\rho} = \inter{t}_{\xT_\xE,\rho}$.
    Hence,  $[s\sigma]_\xE = [t\sigma]_\xE$ by the claim,
    and therefore, $s\sigma \Cab[\xE]{*} t\sigma$.
    ($\Leftarrow$)
    Let $\rho$ be a valuation over $\xT_\xE$
    satisfying the constraints $\varphi$ and 
    $\rho(x) \in \xI(\tau)$ for all $x \in X$.
    Take a substitution $\sigma$ in such a way that $\sigma(x) = v_x$ for each $x \in X$,
    where $v_x \in \Val^\tau$ such that $[v_x]_\xE = \rho(x)$.
    By Lemma~\ref{lem:preservation_of_validity_in_underlying_model}, $\inter{\varphi}_{\xM,\rho} = \m{true}$,
    and thus, $\models_{\xM} \varphi\sigma$ by Lemma~\ref{lem:substitution as interpretation}.
    Hence $s\sigma \Cab[\xE]{*} t\sigma$, and thus $[s\sigma]_\xE = [t\sigma]_\xE$.
    Therefore $\inter{s}_{\xT_\xE,\rho} = \inter{t}_{\xT_\xE,\rho}$ by the claim.
    Item \Bfnum{2} follows from item \Bfnum{1}.
\end{proof}

\newcommand{\ProofValidityOnTermAlgebra}{
\ifthenelse{\boolean{OmitProofs}}{\LemmaValidityOnTermAlgebra*}{}
\begin{proof}
    From our assumptions and Lemmas~\ref{lem:quotient algebra}
    and~\ref{lem:congruence on term algebra},
    it immediately follows that $\xT_\xE$ 
    is a CE-$\langle \Sigma,\xM\rangle$-algebra.
    Furthermore, the item \Bfnum{2} follows from the item \Bfnum{1},
    as $\xE \cec \xE$ clearly holds.
    Thus, it remains to show claim \Bfnum{1} of this lemma.
    Let $\xM = \langle \xI, \xJ \rangle$ and $\xT_\xE = \langle \fI, \fJ \rangle$.
    Let us below also abbreviate $[[t]_\mathrm{c}]_{{\sim}_\xE}$ as $[t]_\xE$.
    We first show the following claim.
    \begin{quote}
        \textbf{Claim}: 
Let $\sigma,\rho$ be a substitution and a valuation on $\xT_\xE$, respectively, 
such that $\rho(x) = [\sigma(x)]_\xE$.
Then, we have $[u\sigma]_\xE = \inter{u}_{\xT_\xE,\rho}$ for any term $u$.
   \begin{proof}[Proof of the claim]
    The proof proceeds by structural induction on $u$.
    The case $u = x \in \xV$ follows as
    $[x\sigma]_\xE = [\sigma(x)]_\xE = \rho(x) = \inter{x}_\rho$.
    Consider the remaining case $u \notin \xV$.
    Suppose $u = f(u_1,\ldots,u_n)$ with $f \in \xF$.
    Then, using the induction hypotheses, 
    we have
    $[f(u_1,\ldots,u_n)\sigma]_\xE 
    =  [f(u_1\sigma,\ldots,u_n\sigma)]_\xE 
    =  \fJ(f)([u_1\sigma]_\xE,\ldots,[u_n\sigma]_\xE )
    =  \fJ(f)(\inter{u_1}_\rho,\ldots,\inter{u_n}_\rho)
    =  \inter{f(u_1,\ldots,u_n)}_\rho$.
   \end{proof}
    \end{quote}
    We now proceeds to show the claim (1).
    ($\Rightarrow$)
    Suppose $\models_{\xT_\xE} \CEqn{X}{s}{t}{\varphi}$.
    That is, $\inter{s}_{\xT_\xE,\rho} = \inter{t}_{\xT_\xE,\rho}$,
    for any valuation $\rho$ over $\xT_\xE$
    satisfying the constraints $\varphi$ and 
    $\rho(x) \in \xI(\tau)$ holds for all $x \in X$.
    To show $\xE \cec \CEqn{X}{s}{t}{\varphi}$,
    let us take an $X$-valued substitution $\sigma$ such that $\models_\xM \varphi\sigma$.
    We are now going to show $s\sigma \Cab[\xE]{*} t\sigma$ holds.
    Take a valuation $\rho$ on $\xT_\xE$ as $\rho(x) = [\sigma(x)]_\xE$.
    Then, because $\sigma$ is $X$-valued, 
    $\rho(x) \in \xI(\tau)$ holds for all $x^\tau \in X$.
    Furthermore, as $\Var(\varphi) \subseteq X$,
    $\models_\xM \varphi\rho$ follows from $\models_\xM \varphi\sigma$.
    Thus, $\inter{s}_{\xT_\xE,\rho} = \inter{t}_{\xT_\xE,\rho}$.
    Hence,  $[s\sigma]_\xE = [t\sigma]_\xE$ by the claim above,
    and therefore, $s\sigma \Cab[\xE]{*} t\sigma$.
    ($\Leftarrow$)
    Suppose $\xE \cec \CEqn{X}{s}{t}{\varphi}$.
    That is, $s\sigma \Cab[\xE]{*} t\sigma$ holds,
    for any $X$-valued substitution $\sigma$ such that $\models_\xM \varphi\sigma$.
    To show $\models_{\xT_\xE} \CEqn{X}{s}{t}{\varphi}$,
    let us take a valuation $\rho$ over $\xT_\xE$
    satisfying the constraints $\varphi$ and 
    $\rho(x) \in \xI(\tau)$ holds for all $x \in X$.
    We now going to show $\inter{s}_{\xT_\xE,\rho} = \inter{t}_{\xT_\xE,\rho}$.
    Since $X \subseteq \xVTh$, we have $\tau \in \xSTh$ for each $x^\tau \in X$.
    Thus, there exists a value $v_x \in \Val^\tau$ such that $[v_x]_\xE = \rho(x) \in \xI(\tau)$.
    Take a substitution $\sigma$ in such a way that $\sigma(x) = v_x$ for each $x \in X$.
    Clearly, $\sigma$ is $X$-valued.
    Furthermore, 
    since $\inter{\varphi}_{\xM,\rho} = \m{true}$ by Lemma~\ref{lem:preservation_of_validity_in_underlying_model},
    we have $\models_{\xM} \varphi\sigma$ by Lemma~\ref{lem:substitution as interpretation}.
    Thus, by our assumption, $s\sigma \Cab[\xE]{*} t\sigma$ holds.
    That is, $[s\sigma]_\xE = [t\sigma]_\xE$,
    and therefore $\inter{s}_{\xT_\xE,\rho} = \inter{t}_{\xT_\xE,\rho}$ by the claim above.
    Therefore, 
    $\xT_\xE \models \CEqn{X}{s}{t}{\varphi}$.
\end{proof}
}
\ifthenelse{\boolean{OmitProofs}}{}{\ProofValidityOnTermAlgebra}

Before proceeding to the completeness theorem,
we connect the two notions related to consistency
(Definitions~\ref{def:consistency} and~\ref{def:consistency w.r.t. values}).

\begin{restatable}{lemma}{LemmaConnectingConsistency}
\label{lem:connecting consistency}
    A CE-theory $\fT$ is consistent if and only if
    it is consistent with respect to values.
\end{restatable}

\begin{proof}[Proof (Sketch)]
($\Rightarrow$)
Suppose that $\fT = \langle \xM, \xE \rangle$ is a CE-theory
over a signature $\Sigma$ and let $\xM = \langle \xI,\xJ \rangle$.
Let $\fM = \langle \fI, \fJ \rangle$ be a CE-$\langle \Sigma,\xM\rangle$-algebra 
such that $\models_{\fM} \xE$.
Suppose $u,v \in \Val^\tau$ with $u \Cab[\xE]{*} v$.
Then $\inter{u}_\fM = \inter{v}_\fM$ by Lemma~\ref{lem:soundness w.r.t. conversion}.
Therefore, by $u,v \in \Val^\tau \cong \xI(\tau) \subseteq \fI(\tau)$,
we have $u = \inter{u}_\xM = \inter{u}_\fM = 
\inter{v}_\fM = \inter{v}_\xM = v$.
($\Leftarrow$)
By Lemma~\ref{lem:validity on term algebra}, $\xT_\xE$ is a model of $\fT$.
This witnesses that $\fT$ is consistent.
\end{proof}

\newcommand{\ProofConnectingConsistency}{
\ifthenelse{\boolean{OmitProofs}}{\LemmaConnectingConsistency*}{}
\begin{proof}
($\Rightarrow$)
Suppose that $\fT = \langle \xM, \xE \rangle$ is a CE-theory
over a signature $\Sigma$ and let $\xM = \langle \xI,\xJ \rangle$.
Suppose that $\fT$ is consistent. 
Then, there exists a CE-$\langle \Sigma,\xM\rangle$-algebra $\fM = \langle \fI, \fJ \rangle$
such that $\models_{\fM} \xE$.
Suppose $u,v \in \Val^\tau$ with $u \Cab[\xE]{*} v$.
By the definition of the CE-$\langle \Sigma,\xM\rangle$-algebras,
we have $\xI(\tau) \subseteq\fI(\tau)$.
Also, by $\models_{\fM} \xE$ and $u \Cab[\xE]{*} v$,
we have $\inter{u}_\fM = \inter{v}_\fM$ by Lemma~\ref{lem:soundness w.r.t. conversion}.
Therefore, by $u,v \in \Val^\tau \cong \xI(\tau) \subseteq \fI(\tau)$,
we have $u = \inter{u}_\xM = \inter{u}_\fM = 
\inter{v}_\fM = \inter{v}_\xM = v$.
($\Leftarrow$)
By Lemma~\ref{lem:validity on term algebra}, $\xT_\xE$ is a model of $\fT$.
This witnesses that $\fT$ is consistent.
\end{proof}
}
\ifthenelse{\boolean{OmitProofs}}{}{\ProofConnectingConsistency}

We now arrive at the main theorem of this section.

\begin{restatable}[completeness]{theorem}{TheoremCompleteness}
\label{thm:completeness}
    Let $\fT = \langle \xM, \xE \rangle$ 
    be a consistent CE-theory.
    Then, we have $\xE \cec \CEqn{X}{s}{t}{\varphi}$
    if and only if 
    $\xE \models \CEqn{X}{s}{t}{\varphi}$.
\end{restatable}


\begin{proof}
The \textit{only if}~part 
follows from 
Theorem~\ref{thm:soundness of algebraic semantics w.r.t. constrained equational validity}.
Thus, it remains to show the \textit{if}~part.
Suppose contrapositively
that $\xE \cec \CEqn{X}{s}{t}{\varphi}$ does not hold.
Then, by Lemma~\ref{lem:validity on term algebra}~\Bfnum{1},
$\not\models_{\xT_\xE} \CEqn{X}{s}{t}{\varphi}$.
Since $\models_{\xT_\xE} \xE$, by Lemma~\ref{lem:validity on term algebra}~\Bfnum{2},
this witnesses that 
there exists a CE-$\langle \Sigma,\xM\rangle$-algebra $\fM$
such that $\models_{\fM} \xE$ but not 
$\models_{\fM} \CEqn{X}{s}{t}{\varphi}$.
This means $\xE \not\models \CEqn{X}{s}{t}{\varphi}$.
This completes the proof of the \textit{if}~part.
\end{proof}

To conclude this section, we explain the postponed question from
the beginning of the section on the definition of CE-algebras. The question
was on why it is required to include those equipped with underlying extended
models---if such models would not be allowed, one does not obtain the completeness
result as witnessed by the following example.

\begin{example}
\label{exp:why underlying model is an extension}
Consider integer arithmetic for the underlying model $\xM$.
Take a term signature 
$\xFTe = \SET{ \m{f}\colon \textsf{Ints} \to \textsf{Bool},
\m{g}\colon \textsf{Ints} \to \textsf{Bool} }$.
Consider the LCES $\xE = 
\SET{
\CEqn{}{\m{f}(x)}{\m{true}}{x \geqslant \m{0}},
\CEqn{}{\m{f}(x)}{\m{true}}{x < \m{0}},
\CEqn{}{\m{g}(x)}{\m{true}}{} 
}$.
By orienting the equations in an obvious way,
we obtain a complete LCTRS (e.g.~\cite{KN13frocos,WM18}).
Then as $\m{g}(x){\downarrow} = \m{true}  \neq 
\m{f}(x) = \m{f}(x){\downarrow}$,
it turns out that no conversions hold between $\m{g}(x)$ and $\m{f}(x)$.
It follows from the Theorem~\ref{thm:relation between validity and conversion of rewriting} that $\xE \not\cec \CEqn{\varnothing}{\m{g}(x)}{\m{f}(x)}{}$.
Now, from Theorem~\ref{thm:completeness}, 
we have $\xE \not\models \CEqn{\varnothing}{\m{g}(x)}{\m{f}(x)}{}$,
i.e.\ one should find a model that witnesses this invalidity.
Indeed,
one can take a CE-$\langle \Sigma,\xM\rangle$-algebra $\fM = \langle \fI, \fJ \rangle$ 
with $\fI(\textsf{Ints}) =  \mathbb{Z} \cup \SET{ \bullet }$,
where $\bullet \notin \mathbb{Z}$,
with the interpretations:
$\fJ(\textsf{f})(\bullet) =  \mathsf{false}$,
$\fJ(\textsf{f})(a) =  \mathsf{true}$ for all $a \in \mathbb{Z}$, and
$\fJ(\textsf{g})(x) =  \mathsf{true}$  for all $x \in \mathbb{Z} \cup \SET{ \bullet }$.
On the other hand, if we would require to take $\fI(\textsf{Ints}) =  \mathbb{Z}$,
then we do not get any model that invalidates this CE.
\end{example}

\section{Related Work}
\label{sec:related-work}

Constrained rewriting began to be popular around 1990, which has been
initiated by the motivation to achieve a tractable solution for
completion modulo equations (such as AC, ACI, etc.), by separating off the
(intractable) equational solving part as constraints. These constraints
mainly consist of (dis)equality of built-in equational theories such as $x*y
\approx_{AC} y*x$. A constrained completion procedure in such a framework
is given in~\cite{KK89}; it is well-known that the specification language Maude
also deals with such built-in theories~\cite{Mes12}. This line of research was
extended to a framework of rewriting with constraints of an arbitrary
first-order formula in~\cite{KK89}, where various completion methods have been
developed for this. However, they, similar to us, mainly considered
term algebras as the underlying models, because the main motivation was to
deal with a wide range of completion problems by separating off some
parts of the equational theory as constraints.

Another well-known type of constraints studied in the context of
constrained rewriting is membership constraints of regular tree
languages. This type of constraints is motivated by dealing with terms
over an (order-)sorted signature and representing an infinite number of
terms that obeys regular patterns obtained from divergence of theorem
proving procedures. In this line of research, \cite{Com98I,Com98II} give a
dedicated completion method for constrained rewrite systems with membership
constraints of regular tree languages. Further a method for
inductive theorem proving for conditional and constrained systems, which is
based on tree grammars with constraints, has been proposed in~\cite{BJ12}.
We also want to mention~\cite{YT89} as a formalism with more abstract
constraints---confluence of term rewrite systems with membership constraints
over an arbitrary term set has been considered there.

The work in this era which is in our opinion closest to the LCTRSs
formalism is the one given in~\cite{DG90}. This is also motivated by giving
a link between (symbolic) equational deduction and constraint solving. Thus,
they considered constraints of an arbitrary theory such as linear integer
arithmetic, similarly to LCTRSs. Based on the initial model of this framework,
they gave an operational semantics of constraint equational logic programming.

The introduction of the LCTRS framework is more recent, and was
initiated by the motivation to deal with built-in data structures such as
integers, bit-vectors etc.\ in order to verify programs written by
real-world programming languages with the help of SMT-solvers. A detailed
comparison to the works in this line of research has been given in~\cite{KN13frocos}.

All in all, to the best of our knowledge, there does not exist anything in
the literature on algebraic semantics of constrained rewriting and Birkhoff
style completeness, as considered in this paper.

\section{Conclusion}
\label{sec:conclusion}

With the goal to establish a semantic formalism of logically constrained
rewriting, we have introduced the notions of constrained equations and
CE-theories. For this, we have extended the form of these constrained
equations by specifying explicitly the variables, which need to be
instantiated by values, in order to treat equational properties in a uniform
way.
Then we have introduced a notion of CE-validity to give a uniform foundation
from a semantic point of view for the LCTRS formalism. After establishing basic
properties of the introduced validity, we have shown the relation to
the conversion of rewriting. Then we presented a sound inference
system $\mathbf{CEC}_0$ to prove validity of constrained equations in
CE-theories. We have demonstrated its ability to establish validity via some
examples. A partial completeness result and a discussion on why only partial
completeness is obtained followed. Finally, we devised sound and complete
algebraic semantics, which enables one to show invalidity of constrained
equations in CE-theories. Furthermore, we have derived an important
characterization of CE-theories, namely, consistency of CE-theories, for which
the completeness theorem holds.
Thus, we have established a basis for CE-theories and their validity by
showing its fundamental properties and giving methods for proving and disproving
the validity of constrained equations in CE-theories.

The question whether there exists a sound and complete proof system that
captures CE-validity remains open. Part of our future work is the
automation of proving validity of constrained equations.

This paper uses the initial formalism of LCTRSs given
in~\cite{KN13frocos}. However, there exists a variant which incorporates the
interpretation of user-defined function symbols by the term algebra~\cite{CL18}.
This variant is incomparable to the initial one. Nevertheless, to investigate
the semantic side of LCTRSs, the initial formalism was a reasonable starting
point. The adaptation of the current work to the extended formalism is also
a part of our future work.



\bibliography{biblio}

\ifthenelse{\boolean{OmitProofs}}{\appendix}{\end{document}}

\section{Omitted Proofs}

In this appendix, omitted proofs are presented.

\subsection{Proofs in Section~\ref{sec:preliminaries}}

\ifthenelse{\boolean{OmitProofs}}{\ProofLemmaSubstitutionAsInterpretation}{}
\ifthenelse{\boolean{OmitProofs}}{\ProofNoRespSubstImplyFalse}{}
\ifthenelse{\boolean{OmitProofs}}{\ProofLemmaPropertiesOfCalculationSteps}{}

\subsection{Proofs in Section~\ref{sec:constrained-equational-validity}}
\ifthenelse{\boolean{OmitProofs}}{\ProofSymAndClosurePropertiesOfEquationalSteps}{}
\ifthenelse{\boolean{OmitProofs}}{\ProofRule}{}
\ifthenelse{\boolean{OmitProofs}}{\ProofCongruence}{}
\ifthenelse{\boolean{OmitProofs}}{\ProofStability}{}
\ifthenelse{\boolean{OmitProofs}}{\ProofStabilityGeneral}{}
\ifthenelse{\boolean{OmitProofs}}{\ProofModelConsequence}{}
\ifthenelse{\boolean{OmitProofs}}{\ProofGeneralTheoryRule}{}
\ifthenelse{\boolean{OmitProofs}}{\ProofRelationBetweenValidityAndConversionOfRewriting}{}
\ifthenelse{\boolean{OmitProofs}}{\ProofProvingValidityByRewriting}{}

\subsection{Proofs in Section~\ref{sec:proving-validity}}
\ifthenelse{\boolean{OmitProofs}}{\ProofPreservationOfCE}{}
\ifthenelse{\boolean{OmitProofs}}{\ProofSoundnessOfTheSystemCEC}{}
\ifthenelse{\boolean{OmitProofs}}{\ProofAbstractingReflexivity}{}
\ifthenelse{\boolean{OmitProofs}}{\ProofAbstractingCalcStep}{}

\subsection{Proofs in Section~\ref{sec:algebraic-semantics}}
\ifthenelse{\boolean{OmitProofs}}{\ProofSoundnessWrtConversion}{}
\ifthenelse{\boolean{OmitProofs}}{\ProofSoundnessOfAlgebraicSemantics}{}
\ifthenelse{\boolean{OmitProofs}}{\ProofQuotientAlgebra}{}
\ifthenelse{\boolean{OmitProofs}}{\ProofTermAlgebra}{}
\ifthenelse{\boolean{OmitProofs}}{\ProofCongruenceOnTermAlgebra}{}
\ifthenelse{\boolean{OmitProofs}}{\ProofValidityOnTermAlgebra}{}
\ifthenelse{\boolean{OmitProofs}}{\ProofConnectingConsistency}{}

\end{document}